\documentclass[11pt, oneside]{article}
\usepackage{graphicx}
\usepackage{amsmath}
\usepackage{amsthm}
\usepackage{amssymb}
\usepackage{tikz}
\usetikzlibrary{shapes}
\usetikzlibrary{patterns}
\usetikzlibrary{arrows,positioning,decorations.pathmorphing,trees}
\usetikzlibrary{arrows.meta}
\usepackage{pgfplots}
\usepackage{float}
\usepackage{fullpage}
\usepackage{hyperref}

\newcommand{\x}{{\bf x}}
\newcommand{\y}{{\bf y}}
\newcommand{\e}{{\bf e}}
\newcommand{\0}{{\bf 0}}
\newcommand{\be}{{\bf b}}

\allowdisplaybreaks[1]

\newcommand{\wel}{{\text{\sc sw}}}
\newcommand{\opt}{{\text{\sc opt}}}
\newcommand{\argopt}{{\text{\sc argopt}}}

\newtheorem{theorem}{Theorem}[section]

\newtheorem{corollary}[theorem]{Corollary}

\newtheorem{lemma}[theorem]{Lemma}
\newtheorem{claim}[theorem]{Claim}

\newtheorem*{remark*}{Remark}
\newtheorem*{notation*}{Notation}
\newtheorem*{observation*}{Observation}
\newtheorem*{theorem*}{Theorem}
\newtheorem*{definition*}{Definition}
\newtheorem*{axiom*}{Axiom}
\newtheorem*{claim*}{Claim}
\newtheorem*{lemma*}{Lemma}

\title{Two-Buyer Sequential Multiunit Auctions\\ with No~Overbidding}
\author{Mete \c{S}eref Ahunbay\thanks{Dept. of Mathematics and Statistics, McGill University: {\tt mete.ahunbay@mail.mcgill.ca}}
\and Brendan Lucier\thanks{Microsoft Research New England: {\tt brlucier@microsoft.com}}  
\and Adrian Vetta\thanks{Dept. of Mathematics \& Statistics and School of Computer Science, McGill University: {\tt adrian.vetta@mcgill.ca}}}

\begin{document}

\maketitle

\begin{abstract}
We study equilibria in two-buyer sequential second-price (or first-price) auctions for identical goods.  
Buyers have weakly decreasing incremental values, and we make a behavioural no-overbidding assumption: the buyers 
do not bid above their incremental values. 
Structurally, we show equilibria are intrinsically linked to a greedy bidding strategy. We then prove three
results.
First, any equilibrium consists of three phases: a competitive phase, a competition reduction phase 
and a monopsony phase. In particular, there is a time after which one buyer exhibits
monopsonistic behaviours.
Second, the declining price anomaly holds: prices weakly decrease over time at any 
equilibrium in the no-overbidding game, a fact previously known for equilibria with overbidding. 
Third, the 
price of anarchy of the sequential auction is exactly $1 - 1/e$.
\end{abstract}

\section{Introduction}\label{sec:intro}

In a two-buyer multiunit sequential auction a collection of $T$ identical items are sold one after another.
This is done using a single-item second-price (or first-price) auction in each time period.
Due to their temporal nature, equilibria in sequential auctions are extremely complex and somewhat 
misunderstood objects~\cite{GS01,PST12,PNV19}. This paper aims to provide a framework
 in which to understand two-buyer sequential auctions. Specifically, we study equilibria in the auction setting 
where both {\em duopsonists} have non-decreasing incremental valuation functions under the natural assumption 
of no-overbidding. Our main technical contribution is an in-depth 
analysis of the relationship between equilibrium bidding strategies and a greedy
behavioural strategy. This similitude allows us to provide a characterization
of equilibria with no-overbidding and to prove three results.

One, any equilibrium in a two-buyer sequential auction with no-overbidding induces three phases:
a competitive phase, a competition reduction phase and a monopsony phase.
In particular, there is a time after which one of the two duopsonists will behave as a {\em monoposonist}.
Here monopsonistic behaviour refers to the type of strategies expected from a buyer with the ability to clinch the 
entire market. Intriguingly, we show that this fact does not hold
for equilibria where overbidding is permitted.

Two, the {\em declining price anomaly} holds for two-buyer sequential auctions with no-overbidding; 
the price weakly decreases over time for any equilibrium in the auction. This result shows that the
seminal result of Gale and Stegeman~\cite{GS01},
showing the declining price anomaly holds for equilibria in two-buyer sequential auctions with overbidding permitted,
carries over to equilibria in auctions with no-overbidding.  Notably, this declining price anomaly can fail to hold for three 
or more buyers, even with no-overbidding~\cite{PNV19}.

Three, the {\em price of anarchy} in two-buyer sequential auctions with no-over\-bidding is exactly $1-\frac{1}{e}\simeq 0.632$.
We remark that the same bound has been claimed in \cite{BBB08,BBB09} for equilibria where overbidding is allowed but, 
unfortunately, there is a flaw in their arguments (see Section~\ref{sec:PoA-overbidding}).

\subsection{Related Work}

The complete information model of two-buyer sequential auctions studied in this paper was introduced by Gale and Stegeman~\cite{GS01}.
This was extended to multi-buyer sequential auctions by Paes Leme et al.~\cite{PST12} (see also~\cite{PNV19}).
Rodriguez~\cite{Rod09} studied equilibria in the special case of identical items and identical buyers with endowments.

Ashenfelter~\cite{Ash89} observed that the price of identical lots fell over time at a sequential
auction for wine. This tendency for a decreasing price trajectory is known as the {\em declining price anomaly}~\cite{MV93}.
Many attempts have been made to explain this anomaly and there is now also a plethora of 
empirical evidence showing its existence in practice; see~\cite{Ashta06,SGL17,PNV19} and the references within for more details. 
On the theoretical side, given complete information, Gale and Stegeman~\cite{GS01} proved that a weakly decreasing price trajectory is guaranteed
in a two-buyer sequential auction for identical items. Prebet et al.~\cite{PNV19} recently proved that declining prices are {\em not}
assured in sequential multiunit auctions with three or more buyers, but gave experimental evidence to show that counter-examples to the
anomaly appear extremely rare.

In the computer science community research has focussed on the welfare of equilibria in sequential auctions.
Bae et al. \cite{BBB08,BBB09} study the {\em price of anarchy} in two-buyer sequential auctions for identical items.
There has also been a series of works bounding the price of anarchy in multi-buyer sequential auctions 
for non-identical goods; see, for example, \cite{PST12,ST13,FLS13}. 

Sequential auctions with incomplete information have also been studied extensively since the
classical work of Milgrom and Weber \cite{MW82,Web83}. We remark that to study the temporal aspects 
of the auction independent of informational aspects it is natural to consider the case of complete information. 
Indeed, our work is motivated by the fact that, even in the basic setting of 
complete information, the simplest case of two-buyers is not well understood.


\subsection{Overview of the Paper}

Section~\ref{sec:examples} presents the model of two-buyer sequential auctions with complete information.
It also includes a collection of examples that illustrate some of 
the difficulties that arise in understanding sequential auctions and provide the reader with a light introduction to 
some of the technical concepts that will play a role in the subsequent analyses of equilibria.
They will also motivate the importance and relevance of no-overbidding. 
This section concludes by incorporating tie-breaking rules in winner determination. 
Section~\ref{sec:greedy} provides a measure for the {\em power of a duopsonist} and presents a natural greedy 
bidding strategy that a buyer with duopsony power may apply.
Section~\ref{sec:technical} studies how prices and duopsony power evolve over time when the buyers apply the greedy bidding strategy. 

The relevance of greedy bidding strategies is exhibited in Section~\ref{sec:equilibria} where we explain 
their close relationship with equilibria bidding strategies. This relationship allows us to provide a characterization of 
equilibria with no-overbidding. Key features of equilibria follow from these structural results.
First, any equilibria induces three distinct phases with a time after which some buyer behaves as a 
monopsonist. Second, prices weakly decrease over time for any equilibrium.
Finally, in Section~\ref{sec:PoA} we prove the price of anarchy is exactly $1-\frac{1}{e}$.
\section{Two-Buyer Sequential Auctions}\label{sec:examples}

In this section we introduce two-buyer sequential auctions and illustrate their strategic aspects via a set 
of simple examples. There are $T$ items to be sold, one per time period by a second-price auction.\footnote{We present our results for 
second-price auctions. Given an appropriate formulation of the bidding space to ensure the
existence of an equilibrium~\cite{PST12} these results also extend to the case of first-price auctions.}
 Buyer~$i\in \{1,2\}$ has a {\em value} $V_i(k)$ for obtaining exactly $k$ items
and {\em incremental value} $v_i(k)=V_i(k)-V_i(k-1)$ for gaining a $k$th item.
We assume $V_i(\cdot)$ is normalised at zero and concave.

\

\noindent{\sc Example 1:}
Consider a two-buyer auction with two items, where the incremental valuations are $(v_1(1), v_1(2))=(10,9)$ and $(v_2(1), v_2(2))=(8,5)$.
The outcome that maximizes social welfare is for buyer~$1$ to receive both copies of the item. However, at equilibrium, 
buyer~$2$ wins the first item at a price of $6$, and buyer~$1$
wins the second item at a price of $5$. To see this, imagine that buyer~$1$ wins in the first period.
Then in the second period she will have to pay $8$ to beat buyer~$2$. Given this, buyer~$2$
will be willing to pay up to $8$ to win in the first round. Thus, buyer~$1$ will win both permits for $8$ each
and obtain a {\em utility} (profit) of $(10+9)-2\cdot 8= 3$. Suppose instead that buyer~$2$ wins in the first round.
Now in the second period, buyer~$1$ will only need to pay $5$ to beat buyer~$2$, yielding a
profit of $10-5=5$. So, by bidding $6$ in the first period, buyer~$1$ can guarantee herself a profit of $5$. 
Given this bid, buyer~$2$ will maximize his own utility by winning the first permit for $6$.
Note that this outcome, the only rational solution, proffers suboptimal welfare.

\subsection{An Extensive-Form Game}
We compactly model this sequential auction as an extensive-form game with complete information
using a directed graph. The node set is given $\mathbb{H} = \{ (x_1, x_2) \in \mathbb{Z}_+ | x_1 + x_2 \leq T \}$.
Each node has a label $\x = (x_1, x_2)$ 
denoting how many items each buyer has currently won.
There is a {\em source node}, $\0=(0,0)$, corresponding to the initial round of the auction, 
and {\em terminal nodes} $(x_1,x_2)$, where $x_1+x_2=T$. If $\x$ is a terminal node, we write $\x \in \mathbb{H}_0$, otherwise 
we say that $\x$ is a \emph{decision node} and write $\x \in \mathbb{H}_+$. We also denote by $t(\y) = T - y_1 - y_2$ the number 
of items remaining to be sold at node $\y$; when the decision node is actually denoted $\x$, we simply write $t$ for $t(\x)$.

We also extend our notation for incremental valuations.
Specifically, we denote the incremental value of buyer~$i$ of a $k$th additional item {\em given} endowment~$\x$ (i.e. from decision node~$\x$) 
as $v_i(k|\x) = V_i(\x_i+k) - V_i(\x_i+k-1)$, for $k \in \mathbb{N}$. For valuations at the source node, we drop the explicit notation of the decision node: 
for example, $v_i(k)=v_i(k| {\bf 0})$.

We find an equilibrium by calculating the {\em forward utility} of each buyer at every node.
The forward utility is the profit a buyer will earn from that period in the auction onwards.
There is no future profit at the end of the auction, so the forward utility of each buyer is zero
at each terminal node. The forward utilities at decision nodes are then obtained by backwards induction on $t$:
each decision node $\x$ has a left child $\x+\e_1$ and a right child $\x+\e_2$, respectively corresponding to
buyer $1$ and $2$ winning an item.
For the case of second-price auctions, it is a weakly dominant strategy for each buyer to bid
its {\em marginal value for winning}. This bid value is the incremental value plus the forward utility of winning
minus the forward utility of losing. Thus, at the node $\x$, the bids of each buyer are
\begin{align*}
b_1(\x) &= v_1(1|\x)+u_1(\x+\e_1) - u_1(\x+\e_2),\\
b_2(\x) &= v_2(1|\x)+u_2(\x+\e_2) - u_2(\x+\e_1).
\end{align*}

\noindent If $b_1(\x) \ge b_2(\x)$ then buyer~$1$ will win and the forward utilities at $\x$ are then
\begin{align*}
u_1(\x) &= v_1(1|\x)+u_1(\x+\e_1) - b_2(\x+\e_2)\\
&= \left( v_1(1|\x)- v_2(1|\x)\right) + u_1(\x+\e_1) - u_2(\x+\e_2) +u_2(\x+\e_1), \\
u_2(\x) &= u_2(\x+\e_1).
\end{align*}
The forward utilities are defined symmetrically  if $b_1(\x) \le b_2(\x)$ and buyer~$2$ wins. 
Given the forward utilities at every node, the iterative elimination of
weakly dominated strategies then produces a unique equilibrium~\cite{GS01,BBB08}.

The auction of Example 1 is illustrated in Figure~\ref{fig:2-buyer-second-price}.
The first row in each node contains its label $\x=(x_1, x_2)$ and also the number of items, $t=T-x_1-x_2$, 
remaining to be sold. The second row shows the forward utility of each buyer.
Arcs are labelled by the bid value; here arcs for buyer~$1$ point left and arcs for buyer~$2$ point right.
Solid arcs represent winning bids and dotted arcs represent losing bids. The equilibrium path, in bold,
verifies our previous argument: buyer~$2$ wins the first item at price $6$ and buyer~$1$
wins the second item at price $5$.

\begin{figure}[h]
\centering
 \begin{minipage}{6cm}
  \begin{tikzpicture}[scale=0.55]
  \node [ellipse,draw, fill=red!20, align=center, scale=0.5] (v1) at (4,4) { {\bf (0,0)--2} \\ {\bf 5} : {\bf 2} };
\node [ellipse,draw, fill=red!20, align=center, scale=0.5](v2) at (2,2) { {\bf(1,0)--1} \\ {\bf 1} : {\bf 0} };
\node [ellipse,draw, fill=red!20, align=center, scale=0.5](v3) at (6,2) { {\bf(0,1)--1} \\ {\bf 5} : {\bf 0} };
\node[ellipse,draw, fill=yellow!20, align=center, scale=0.5](v4) at (0,0) { {\bf (2,0)--0} \\ {\bf 0} : {\bf 0}};
\node[ellipse,draw, fill=yellow!20, align=center, scale=0.5](v5) at (4,0) { {\bf (1,1)--0} \\  {\bf 0} : {\bf 0} };
\node[ellipse,draw, fill=yellow!20, align=center, scale=0.5](v6) at (8,0) { {\bf (0,2)--0} \\ {\bf 0} : {\bf 0} };
\draw [->, dotted] (v1) -- (v2) node[very near start,left, scale = .66]{$6\ $};
\draw [->, very thick] (v1) -- (v3) node[very near start,right, scale = .66]{$\ \ 8$};
\draw [->] (v2) -- (v4) node[very near start,left, scale = .66]{$9\ $} ;
\draw [->, dotted] (v2) -- (v5) node[very near start,right, scale = .66]{$\ 8$};
\draw [->, very thick] (v3) -- (v5) node[very near start,left, scale = .66]{$10\ $};
\draw [->, dotted] (v3) -- (v6) node[very near start,right, scale = .66]{$\ \ 5$};
\end{tikzpicture}
 \end{minipage}
\caption{The extensive form for Example 1. The set of histories has the structure of a rooted tree. Iteratively solving for an equilibrium gives
the auction tree, with bidding strategies and forward utilities shown.}
\label{fig:2-buyer-second-price}
\end{figure}
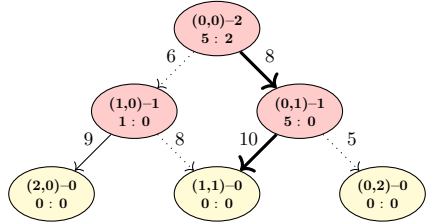

Consequently, in a two-buyer sequential auction, each subgame corresponds to a standard {\em second-price auction}.
We remark that for sequential auctions with three or more buyers each subtree in the extensive-form game 
corresponds to an {\em auction with interdependent valuations} 
(or an {\em auction with externalities}) \cite{Fun96,JM96}.
As a result, equilibria in such multi-buyer sequential auctions are even more complex than for two buyers;
see \cite{PST12,PNV19} for details.

\subsection{No-Overbidding}
Unfortunately, equilibria in sequential auctions can have undesirable and unrealistic properties. In particular they 
may exhibit severe {\em overbidding}.

\

\noindent{\sc Example 2:}
Consider a sequential auction with $T$ items for sale and incremental valuations shown in Figure~\ref{fig:large-overbidding}, 
where $0 < \epsilon \ll 1/T^2$. 
The key observation here is that if buyer~$1$ wins her first item before the final period she will then win in every 
subsequent period for a price $1-\epsilon$.
On the other hand, if buyer~$1$ wins her first item in the final period then the price will be $0$. This is because 
buyer~$2$ must then have won the first $T-1$ items and so has no value for winning in the final round.

\begin{figure}[h!]
\centering
\begin{tabular}{|c|c|c|c|c|c|}
\hline
Incremental Values  & $v_i(1)$ & $v_i(2)$  & $\dots$ & $v_i(T-1)$ & $v_i(T)$ \\ [0.5ex] 
\hline
{\sc Buyer 1}& $1$ & $1$& $\dots$  &$1$ & $1$ \\
\hline
{\sc Buyer 2}& $1-\epsilon$ &  $1-\epsilon$ & $\dots$& $1-\epsilon$ & 0\\
\hline
\end{tabular}
\caption{Incremental values of each buyer which induce ``severe'' overbidding.}
\label{fig:large-overbidding}
\end{figure}

These observations imply that buyer~$1$ will bid $b_1(t)=(T-t)\cdot \epsilon$ in period $t$.
At equilibrium, buyer~$2$ will beat these bids in the first $T-1$ periods and make a profit of
$(1-\epsilon)\cdot (T-1)-\frac12 (T-1)T\cdot \epsilon =\Omega(T)$. But if buyer~$2$ loses the first item, he will win no items 
at all in the auction and thus his forward utility from losing is zero.
Consequently, his marginal value for winning the first period is $\Omega(T)$ and so he will bid $b_2(1)=\Omega(T)\gg 1-\epsilon$.
Thus, at the equilibrium, buyer~$2$ will massively overbid in nearly every round.

\

Overbidding in a sequential auction is very risky and depends crucially on perfect information, so it is rare in practice.
To understand some of these risks consider again Example~$2$.  Equilibria are very sensitive to the valuation functions 
and any {\em uncertainty} concerning the payoff valuations could lead to major changes in the outcome. For instance, suppose
buyer~$1$ is mistaken in her belief regarding the $T$th incremental value of buyer~$2$. Then she will be unwilling to 
let buyer~$2$ win the earlier items at a low price. Consequently, if buyer~$2$ bids $b_2(1)=\Omega(T)$ then he will make a loss,
and continuing to follow the equilibrium strategy will result in a huge loss. This is important even with complete
information because, for computational or behavioral reasons, a buyer cannot necessarily assume with certainty that the other buyer will follow the equilibrium 
prescription; for example, the computation of equilibria in extensive-form games is hard.  Likewise in competitive settings with externalities, where the 
a buyer may have an interest in limiting the profitability of its competitor, overbidding is an unappealing option.
We address this wedge between theory and practice by imposing a non-overbidding assumption, and indeed
such assumptions are common in the theoretical literature \cite{ST13,CKS16}. We leave the analysis of models that explicitly 
capture the risks described above as a direction for future research.

For our sequential auctions, given its valuation function, each buyer will naturally constrain its bid by its incremental value.
So we will assume this {\em no-(incremental) overbidding} property:
\begin{equation}\label{con:no-overbidding}
b_i(\x) \le v_i(1|\x)
\end{equation}
In particular, at each stage a buyer will bid the {\em minimum} of its incremental value and its marginal value for winning.

We note that the no-overbidding property is especially well-suited to our setting of valuations that exhibit decreasing marginal values and free disposal.  That is, valuations that are non-decreasing and weakly concave.  Without these assumptions, sequential auctions can exhibit severe {\em exposure problems}\footnote{The exposure problem arises when a buyer has large value for a set $S$ of items
but much less value for strict subsets of $S$. Thus bidding for the items of $S$ sold early in the auction exposes the buyer to
a high risk if he fails to win the later items of~$S$.} that introduce inefficiencies driven by the tension of overbidding.  For this reason, sequential auctions are pathologically inappropriate mechanisms when valuations are not concave or monotone.  Many practical sequential multiunit auctions therefore assume (or impose) that buyers declare concave non-decreasing valuations. 
For example, in cap-and-trade (sequential) auctions, 
such as the Western Climate Initiative (WCI) and the Regional Greenhouse Gas Initiative (RGGI),
multiple items are sold in each time period but bids are constrained to be weakly decreasing.  We follow the literature on multiunit 
sequential auctions and focus on concave and non-decreasing valuations, where a no-overbidding constraint is more natural.

\subsection{Tie-Breaking Rules}

When overbidding is allowed the forward utilities at equilibria are unique (see also~\cite{GS01}), regardless of the
tie-breaking rule.
Surprisingly, this is {\bf not} the case when overbidding is prohibited:

\ 

\noindent{\sc Example 3:}
Take a four round auction, where $v_i(k) = 1$ for $k \leq 3$ and $v_i(k) = 0$ otherwise. 
Solving backwards, the forward utilities are the same for every {\em non-source node}
whether or not overbidding is permitted. 
In particular, at the successor nodes of the source we have $u_i(\e_i) = 2$ and $u_i(\e_{-i}) = 1$. 
But now a difference occurs. Without the overbidding constraint, 
both buyers would bid $2$ at the source node $\0$ and, regardless of the winner, each buyer has
$u_i(\0)=1$. But with the no-overbidding constraint both buyers will bid $1$. Consequently, if this tie is 
broken in favour of buyer~$1$ with probability $p$, then buyer~$1$ has forward 
utility $u_1(\0) = 1+p$ and buyer~$2$ has forward utility $u_2(\0) = 2-p$, so buyers' payoffs depend on $p$.

\

Thus, under no-overbidding we need to account for the tie-breaking process.
To do this, let $\be=(b_1,b_2)$ where $b_i: \mathbb{H}_+ \rightarrow \mathbb{R}$ is the bidding strategy
of buyer~$i$. Given the bids at the node $\x$, let $\pi_i(\be|\x)$ denote the probability buyer~$i$ 
is awarded the item, where $\pi_i(\be|\x) = 1$ if $b_i(\x) > b_{-i}(\x)$. This defines a tie-breaking rule at each node, 
and the {\em forward utility} of each buyer can again be calculated 
inductively. For any terminal node $\x \in \mathbb{H}_0$ the forward utility is zero: $u_i(\be | \x) = 0$. 
The forward utility of buyer~$i$ at decision node $\x \in \mathbb{H}_+$ is then:
\begin{equation*}
u_i(\be |\x) = \pi_i(\be | \x) \cdot (v_i(1|\x) - b_{-i}(\x) + u_i(\be | \x+\e_i)) + (1 - \pi_i(\be |\x)) \cdot u_i(\be | \x + \e_{-i})
\end{equation*}

With the tie-breaking rule defined, the (expected) forward utilities at equilibria are thus unique.
Moreover, there is a unique bidding strategy $\be$ which survives the iterated elimination of weakly 
dominated strategies. Specifically, under the no-overbidding condition, at any node $\x$ each bidder should bid the minimum of its marginal value for winning and its incremental value:
\begin{equation} \label{eq:bid-value}
 b_i(\x) = \min \, \big[ \, v_i(1|\x)\, ,\, v_i(1|\x) + u_i(\be |\x+\e_i) - u_i(\be | \x+\e_{-i})\, \big]
\end{equation}

\section{Greedy Bidding Strategies}\label{sec:greedy}
To understand equilibria in two-buyer sequential auctions with no-overbidding, 
we need to consider greedy bidding strategies.
At decision node $\x$, suppose buyer~$i$ attempts to win exactly $k$ items by the following strategy:
she waits (bids zero) for $t-k$ rounds and then outbids buyer~$-i$ in the final $k$
rounds. For this, by the no-overbidding assumption, she must bid
$\geq v_{-i}(1|(t-k)\cdot \e_{-i})$ in the final $k$ rounds. If this strategy to be
implementable, it must be that:
$$v_i(k|\x) \ \ge\ v_{-i}(t-k+1|\x)$$
This strategy would then give buyer~$i$ a utility of:
\begin{equation}\label{def:mu-bar}
\bar{\mu}_i(k|\x) \ =\ \sum_{j = 1}^{k} v_i (j |\x) - k \cdot v_{-i}(t - k + 1 |\x) 
\end{equation}
If buyer~$i$ attempts to apply this greedy strategy, it should select $k$ to maximize its profit $\bar{\mu}_i(k|\x)$.
So in equilibrium, buyer $i$ should earn at least the maximum of these utilities over all feasible $k$. Remarkably, this property 
need {\bf not} be true for equilibria when overbidding is allowed; see Example~4 below.

Buyer $i$'s \emph{greedy utility} at decision node $\x$ is the resultant utility from applying its greedy strategy from $\x$:
\begin{equation}\label{def:minmax-utility}
\mu_i(\x) \ =\  \max_{k \in [t] \cup \{0\}} \bar{\mu}_i(k|\x) 
\end{equation}
In turn, buyer~$i$'s corresponding \emph{greedy demand} at $\x$ is:
\begin{equation}\label{def:minmax-demand}
\kappa_i(\x) \ =\  \min \arg \max_{k \in [t] \cup \{0\}} \bar{\mu}_i(k|\x) 
\end{equation}
But when can buyer~$i$ profitably apply this greedy strategy? It can apply it whenever it has {\em duopsony power}.
In a sequential auction this ability arises when $v_i(1|\x) \ >\ v_{-i}(t|\x)$. 
Formally, let  buyer~$i$'s {\em duopsony factor} at $\x$ be:
\begin{equation}\label{def:duopsony-factor}
f_i (\x) \ =\  \max\{ k \in [t]\, :\, v_i (k |\x) > v_{-i}(t-k+1 |\x)\} \cup \{0\}
\end{equation}
Observe that if $f_i (\x)=0$ then 
buyer~$i$ cannot apply the greedy strategy, and we then have $\mu_i(\x) = 0$ and $\kappa_i(\x) = 0$. On the other hand,
if $f_i(\x) > 0$ then $\mu_i(\x) > 0$, and any maximizer of $\mu_i(\x)$ is necessarily at most $f_i(\x)$. 

We say that a buyer is a {\em monopsonist} if the other buyer has no duopsony power, that is, if $f_{-i}(\x) = 0$. 
In turn, a buyer is a {\em strict monopsonist} if she has total duopsony power, i.e. $f_{i}(\x) = t$.
So in a sequential auction with no-overbidding, a strict monopsonist can guarantee it gains at least its
greedy utility. This is analogous to the corresponding {\em static market} setting.
However, this simple fact can fail to hold when overbidding occurs:

\

\noindent{\sc Example 4:}
Consider a three-item auction, where $v_1(k) = 1$ for any $k \in \{1,2,3\}$, and buyer~$2$ 
has incremental valuations $\big(v_2(1),v_2(2),v_2(3)\big) = (2/3-\delta,1/2+\epsilon,0)$, 
where we fix $\epsilon,\delta > 0$ to be small with $2\epsilon=3\delta$.
With overbidding permitted, in equilibrium with ties broken in favour of buyer~$2$,
 buyer~$1$ wins a single item.

 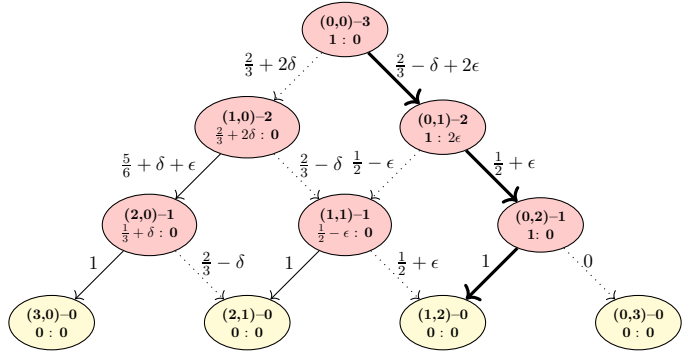
\begin{figure}[h]
 \centering
  \begin{tikzpicture}[scale=.65]
    \node [ellipse,draw, fill=red!20,align=center, scale=0.5] (a1) at (2,6) { {\bf (0,0)--3} \\ {\bf 1} : {\bf 0} };
   \node [ellipse,draw, fill=red!20, align=center, scale=0.5] (b2) at (4,4) { {\bf (0,1)--2} \\ {\bf 1} : {\bf $2\epsilon$} };
  \node [ellipse,draw, fill=red!20, align=center, scale=0.5] (b1) at (0,4) { {\bf (1,0)--2} \\ {\bf $\frac23+2\delta$} : {\bf 0} };
\node [ellipse,draw, fill=red!20, align=center, scale=0.5](c2) at (2,2) { {\bf(1,1)--1} \\ {\bf $\frac12-\epsilon$} : {\bf 0} };
\node [ellipse,draw, fill=red!20, align=center, scale=0.5](c3) at (6,2) { {\bf(0,2)--1} \\ {\bf 1}: {\bf 0} };
  \node [ellipse,draw, fill=red!20,align=center, scale=0.5] (c1) at (-2,2) { {\bf (2,0)--1} \\ {\bf $\frac13+\delta$} : {\bf 0} };
      \node [ellipse,draw, fill=yellow!20, align=center, scale=0.5] (d1) at (-4,0) { {\bf (3,0)--0} \\ {\bf 0} : {\bf 0} };
      \node[ellipse,draw, fill=yellow!20, align=center, scale=0.5](d2) at (0,0) { {\bf (2,1)--0} \\ {\bf 0} :  {\bf 0}};
\node[ellipse,draw, fill=yellow!20, align=center, scale=0.5](d3) at (4,0) { {\bf (1,2)--0} \\  {\bf 0} :  {\bf 0} };
\node[ellipse,draw, fill=yellow!20, align=center, scale=0.5](d4) at (8,0) { {\bf (0,3)--0} \\ {\bf 0} :  {\bf 0} };
\draw [->, dotted] (a1) -- (b1) node[near start,left, scale = .66]{$\ \frac23+2\delta\ $};
\draw [->, very thick] (a1) -- (b2) node[near start,right, scale = .66]{$\ \frac23-\delta+2\epsilon\ $} ;
\draw [->] (b1) -- (c1) node[near start,left, scale = .66]{$\ \frac56+\delta+\epsilon\ $} ;
\draw [->, dotted] (b1) -- (c2) node[near start,right, scale = .66]{$\ \frac23-\delta\ $} ;
\draw [->, dotted] (b2) -- (c2) node[near start,left, scale = .66]{$\ \frac12-\epsilon\ $} ;
\draw [->, very thick] (b2) -- (c3) node[near start,right, scale = .66]{$\ \frac12+\epsilon\ $};
\draw [->] (c1) -- (d1) node[near start,left, scale = .66]{$1\ $}; 
\draw [->, dotted] (c1) -- (d2) node[near start,right, scale = .66]{$\ \frac23-\delta\ $};
\draw [->] (c2) -- (d2) node[near start,left, scale = .66]{$1\ $}; 
\draw [->, dotted] (c2) -- (d3) node[near start,right, scale = .66]{$\ \frac12+\epsilon\ $}; 
\draw [->, very thick] (c3) -- (d3) node[near start,left, scale = .66]{$1\ $}; 
\draw [->, dotted] (c3) -- (d4) node[near start,right, scale = .66]{$0\ $}; 
\end{tikzpicture}
\caption{A sequential auction with overbidding permitted where neither buyer exhibits monopolistic behaviours.}
\label{fig:overbidding}
\end{figure}

Figure~\ref{fig:overbidding} illustrates this example. The key observation here is that $b_2(\0)=2/3 - \delta + 2\epsilon> 2/3 - \delta = v_2(1)$, so buyer $2$ overbids at the source node. Furthermore, buyer~$1$ obtains a profit of $1$ in
this equilibrium with overbidding, but  
$\bar\mu_1(3|\0) = 1+3\delta$. So in the equilibrium with overbidding, buyer~$1$ obtains less 
than her greedy utility.
In contrast, under no-incremental overbidding, buyer~$1$ will win all three items and
make exactly her greedy utility.

\

The greedy strategy induces two types of ``price'' that will be important.
First, we say that the {\em baseline price} of buyer~$i$ at decision node $\x$ is:
\begin{equation}\label{def:baseline-price}
\beta_i(\x,t) \ =\  	\begin{cases}
		v_i (1|\x) & f_i(\x) = 0 \\
		v_{-i}(t - \kappa_i(\x) + 1|\x) & f_i(\x) > 0
		\end{cases} 
\end{equation}
Second, the {\em threshold price} of  buyer~$i$ at decision node 
$\x$ is: 
\begin{equation}\label{def:threshold-price}
p_{i}(\x) \  = \  v_i (1|\x) + \mu_i(\x+\e_i) - \mu_i(\x+\e_{-i})
\end{equation}

The baseline price may be seen as the price a greedy buyer would post if it wanted to obtain its 
greedy utility. By posting a bid of $\beta_i(\x) + \epsilon$ on each node following $\x$, buyer~$i$ would be guaranteed,
by the no-overbidding condition, to win at least $\kappa_i(\x)$ items. The threshold price, in turn, arises from a 
behavioural rule: it is the bid a buyer would make on the assumption that it wins \textbf{exactly} its greedy utility
through the rest of the auction.
%
%

\section{Greedy Bidding Outcomes}\label{sec:technical}

Now imagine that buyers attempt to bid their threshold prices at each decision node, subject to the no-overbidding constraint. 
By definition (\ref{def:threshold-price}) of threshold prices, this corresponds to the behavioural rule where buyers perceive their forward 
utilities to equal their greedy utilities, and bid accordingly. We will discover in Section~\ref{sec:equilibria} that such greedy bidding strategies 
are in some circumstances equivalent to equilibrium bidding strategies under the {\em no-overbidding assumption}.

Accordingly, to understand equilibria we must study the consequences of greedy bidding.
So, in this section, we will inspect the properties of greedy outcomes. 

\subsection{Evolution of the Duopsony Factor}

Recall the {\em duopsony factor} denotes the maximum number of items a buyer can target using the
greedy bidding strategy -- we first inspect its evolution:

\begin{lemma}\label{lem:sdf}
Suppose that $\x$ is a decision node with $t > 1$, then $\forall i \in \{1,2\}$:
\begin{align*}
f_i(\x+\e_i) & = \min \{f_i(\x)-1, 0\} \\
f_i(\x+\e_{-i}) & = \min \{f_i(\x), t-1\}
\end{align*}
\end{lemma}

 To prove Lemma \ref{lem:sdf} we break the statement up to three claims which we prove separately. The first claim shows that a buyer will no duopsony power will remain so for the rest of the auction.

\begin{claim}\label{cl:sdf-0}
If buyer~$i$ has no duopsony power at decision node $\x$, where $t >1$, then it has no duopsony power at any
subsequent node. Formally, for any $j\in \{1,2\}$, 
$$f_i(\x) = 0 \ \Longrightarrow \ f_i(\x+\e_j) = 0$$
\end{claim}
\begin{proof}
Let $f_i(\x) = 0$. Then, by definition~(\ref{def:duopsony-factor}) of the duopsony factor, for any $k\le t$, 
we have $v_{-i}(t-k+1|\x) \geq v_i(k|\x)$. Hence:
\begin{align}\label{eq:v-i}
v_{-i}(t-1-k+1|\x+\e_{-i}) &= v_{-i}(t-k+1|\x) \nonumber \\
& \geq v_i(k|\x) \nonumber  \\
&= v_i(k|\x+\e_{-i})
\end{align}
Here the first equality holds because $v_{-i} (j | \x+\e_{-i}) =v_{-i} (j+1 | \x)$ for any $j \in \mathbb{Z}$.
The second equality holds because the incremental values of buyer~$i$ are independent of
the allocation to the buyer~$-i$ and so $v_i (k | \x+\e_{-i}) =v_i (k| \x)$. Because $t-1< t$, it follows from (\ref{eq:v-i}) that $f_i(\x+\e_{-i}) = 0$. 

Similarly, for any $k \le t$, we have:
\begin{align*}
v_{-i}(t-1-k+1|\x+\e_{i}) &= v_{-i}(t-1-k+1|\x) \\
& \geq v_i(k+1|\x) \\
&= v_i(k|\x+\e_{i}) 
\end{align*}
Therefore $f_i(\x+\e_{i},t-1) = 0$.
\end{proof}

Recall, a buyer~$i$ is a strict monopsonist at $\x$ if $f_i(\x) = t$. The next claim shows that a strict monopsonist remains a strict monopsonist throughout the auction.
\begin{claim}\label{cl:sdf-t}
If buyer~$i$ is a strict monopsonist at decision node $\x$, where $t >1$, then it is
a strict monopsonist at every subsequent node. Formally, for any $j\in \{1,2\}$, 
$$f_i(\x) = t \ \Longrightarrow \ f_i(\x+\e_j) = t-1$$
\end{claim}
\begin{proof}
Suppose $f_i(\x) = t$, then $v_i(t|\x) > v_{-i}(1|\x)$. Then:
\begin{equation*}
v_{i}(t-1|\x+\e_{i}) \ =\  v_{i}(t|\x) \ >\  v_{-i}(1|\x) \ =\  v_{-i}(1|\x+\e_{i})
\end{equation*}
Here the first and second equalities result from considering the shift $\x \rightarrow \x+\e_i$ on the valuations.
The strict inequality results from the assumption $f_i(\x) = t$. So, $f_i(\x+\e_i) = t-1$. Similarly:
\begin{align*}
v_i(t-1|\x+\e_{-i}) & = v_i(t-1|\x) \\
&\geq v_i(t|\x) \\
&> v_{-i}(1|\x) \\
&\geq v_{-i}(2|\x)\\
&= v_{-i}(1|\x+\e_{-i})
\end{align*}
Here, the first and last equalities result from considering the shift $\x \rightarrow x+\e_{-i}$ on the valuations. 
The inequalities result from 
the assumption of non-increasing incremental values, and the strict inequality follows from 
the assumption $f_i(\x) = t$. So $f_i(\x+\e_{-i}) = t-1$.
\end{proof}

The third claim examines the intermediate case in which a buyer has duopsony power but is not a strict monopsonist.
\begin{claim}\label{cl:sdf<t}
Let buyer~$i$ have duopsony power but not total duopsony power at decision node $\x$, where $t>1$. 
If it wins the current item then the maximum number of items to which it can apply the greedy bidding strategy
falls by exactly one; it if loses the current item then the maximum number of items to which it can apply the greedy bidding strategy
remains the same. Formally,
$$
0 < f_i(\x) < t \ \Longrightarrow \ 
\begin{cases}
\ f_i(\x+\e_i) &=\ f_i(\x) - 1\\
\ f_i(\x+\e_{-i}) &=\ f_i(\x)
\end{cases}
$$
\end{claim}
\begin{proof}
Note first that $v_i(k|\x+\e_{-i}) = v_i(k|\x)$ and $v_{-i}((t-1)-k+1|\x+\e_{-i}) = v_{-i}(t-k+1|\x)$ by considering the 
shift in incremental valuations $\x \rightarrow \x + \e_{-i}$. Since $0 < f_i(\x) \leq t-1$, 
for $k \leq f_i(\x), v_i(k|\x) > v_{-i}(t-k+1|\x)$ and for $k > f_i(\x), v_i(k|\x) \leq v_{-i}(t-k+1|\x)$. 
Therefore, $f_i(\x+e_{-i}) = f_i(\x)$.

Similarly, $v_i(k|\x+\e_{i}) = v_i(k+1|\x)$ and $v_{-i}((t-1)-k+1|\x+\e_{i}) = v_{-i}(t-k|\x)$, by considering the shift in 
incremental valuations $\x \rightarrow \x + \e_i$. Now $0 < f_i(\x) \leq t-1$.
Thus, by definition~(\ref{def:duopsony-factor}) of the duopsony factor,
for $k + 1 \leq f_i(\x)$ we have $v_i(k+1|\x) > v_{-i}(t-k|\x)$ and for $k + 1 > f_i(\x)$ we have $v_i(k+1|\x) \leq v_{-i}(t-k|\x)$. 
Therefore, $f_i(\x+e_{i}) = f_i(\x)-1$.
\end{proof}

\subsubsection{Evolution of Greedy Utilities}

Next, we analyse the evolution of buyers' greedy utilities. We find that some buyer $i$'s greedy utility will decrease upon the other buyer's win if and only if she is a strict monopsonist who demands the entire supply. In turn, we bound below the decrease in buyer $i$'s greedy utility if she has duopsony power and wins an item.

\begin{lemma}\label{lem:evolutionMinmaxUtility}
The greedy utility of a buyer weakly decreases when the buyer loses.
Specifically, for any decision node $\x$ and any buyer $i$,
\begin{align*}
\mu_i(\x+\e_{-i}) &= \mu_i(\x) &\mathrm{if}\  \kappa_i(\x)  &< t \\
\mu_i(\x+\e_{-i}) &< \mu_i(\x)  &\mathrm{if}\  \kappa_i(\x) &=t 
\end{align*}
In turn, if buyer~$i$ has non-zero greedy demand at decision node $\x$ and wins an item, then its greedy utility 
decreases by at most the value it would have for purchasing an item at his baseline price. Formally, $\forall \x \in \mathbb{H}_+, \forall i \in \{1,2\}$:
$$ \kappa_i(\x) > 0 \Rightarrow \mu_i(\x+\e_i) \geq \mu_i(\x) - v_i(1|\x) + \beta_i(\x) $$
\end{lemma}

To prove Lemma \ref{lem:evolutionMinmaxUtility}, we will first prove two claims about the payoff
of the greedy strategy targeting the purchase of $k$ items.
Our first claim states that the greedy utility does not change if the buyer loses the first item,
provided its greedy strategy was not attempting to win every item.
\begin{claim}\label{cl:lose-utility-k}
At decision node $\x$, the utility to a buyer from targetting less than $t$ items using the greedy strategy 
remains constant when the buyer loses.
Specifically, for any buyer $i$ and any $k\le t-1$,
\begin{equation*}
\bar{\mu}_i(k|\x+\e_{-i}) \ =\  \bar{\mu}_i(k|\x)
\end{equation*}
\end{claim}
\begin{proof}
By definition~(\ref{def:mu-bar}) of $\bar{\mu}$, we have
\begin{align*}
\bar{\mu}_i(k|\x+\e_{-i}) & = \sum_{j=1}^k v_i(j|\x+\e_{-i}) - k \cdot v_{-i}(t-1-k+1|\x+\e_{-i}) \nonumber \\
& = \sum_{j=1}^k v_i(j|\x+\e_{-i}) - k \cdot v_{-i}(t-k|\x+\e_{-i}) \nonumber \\
& = \sum_{j=1}^k v_i(j|\x) - k \cdot v_{-i}(t-k|\x+\e_{-i}) \nonumber \\
& = \sum_{j=1}^k v_i(j|\x) - k \cdot v_{-i}(t-k+1|\x) \nonumber  \\
& = \bar{\mu}_i(k|\x)
\end{align*}
Here, the first and last equality hold by the definition of $\bar{\mu}$, the second equality holds by simplifying 
the expression in $v_i(\cdot|\x+\e_i)$, and third and fourth equalities hold by considering the shift in valuations 
$\x \rightarrow \x + \e_{-i}$.
\end{proof}

In contrast, if the buyer targets $k$ items using the greedy bidding strategy then its utility can
change if it successfully wins the first item.
\begin{claim}\label{cl:win-utility-k}
Let $\x$ be a decision node where $t > 1$. Then for any buyer $i$ and any $k\le t$
\begin{equation*}
\bar{\mu}_i(k-1|\x+\e_i) \ =\  \bar{\mu}_i(k|\x) - v_i(1|\x) + v_{-i}(t-k+1|\x)
\end{equation*}
\end{claim}
\begin{proof}
By definition~(\ref{def:mu-bar}) of $\bar{\mu}$, we have
\begin{align*}
\bar{\mu}_i(k-1|\x+\e_i) & = \sum_{j=1}^{k-1} v_i(j|\x+\e_i) - (k-1)\cdot v_{-i}((t-1)-(k-1)+1|\x+\e_i) \\
& = \sum_{j=1}^{k-1} v_i(j+1|\x) - (k-1)\cdot v_{-i}(t-k+1|\x+\e_i) \\
& = \sum_{j=1}^{k-1} v_i(j+1|\x) - (k-1) \cdot v_{-i}(t-k+1|\x) \\
& = \sum_{j=2}^{k} v_i(j|\x) - (k-1) \cdot v_{-i}(t-k+1|\x) \\
& = \bar{\mu}_i(k|\x) - v_i(1|\x) + v_{-i}(t-k+1|\x)
\end{align*}
Again, the first and last equalities hold by the definition of $\bar{\mu}$. The second equality follows by simplifying the 
expression in $v_i(\cdot|\x+\e_i)$, the third equality holds by considering the shift in incremental 
valuations $\x \rightarrow \x + \e_i$, and the fourth equality holds by a change of variables.
\end{proof}

Let us now prove Lemma~\ref{lem:evolutionMinmaxUtility}. 
We want to show that, at any decision node $\x$, if a buyer~$i$ greedy demand corresponds 
to the entire set of items for sale, then its greedy utility strictly decreases if buyer~$-i$ wins an item at $\x$. If 
instead buyer~$i$'s greedy demand does not correspond to all the items up for sale, then buyer~$i$'s greedy utility 
does not change if buyer~$-i$ wins an item at $\x$.

\begin{proof}[Proof (of Lemma~\ref{lem:evolutionMinmaxUtility})]
By definition~(\ref{def:minmax-utility}) of the greedy utility and by Claim~\ref{cl:lose-utility-k}, we have
\begin{align}\label{eq:mu-1}
\mu_i(\x+\e_{-i}) 
& = \max_{k \in [t-1] \cup \{0\}} \bar{\mu}_i(k|\x+\e_{-i}) \nonumber \\
& = \max_{k \in [t-1] \cup \{0\}} \ \bar{\mu}_i(k|\x) \nonumber \\
&\leq  \max_{k \in [t] \cup \{0\}} \bar{\mu}_i(k|\x)  \nonumber\\
&=  \mu_i(\x)
\end{align}
But equality holds in the second inequality if and only if $\kappa_i(\x) < t$.
The first part of the lemma then follows immediately from (\ref{eq:mu-1}).

Observe that if buyer~$i$ wins at decision node $\x$ and then wins $k-1$ items implementing its greedy strategy, then
its profit is exactly $\bar\mu_i(k-1|\x+\e_{i}) +v_i(1|\x) - b_{-i}(\x)$. By Claim~\ref{cl:win-utility-k}, this is greater than
$\bar{\mu}_i(k|\x)$ if the other buyer bid $b_{-i}(\x)<v_{-i}(t-k+1|\x)$. Setting $k = \kappa_i(\x)$ shows the 
second part of the lemma.
\end{proof}

\subsubsection{Evolution of Greedy Demand}

Next, we turn attention to how the demand evolves. We show that, if the greedy demand of a buyer is less than the entire supply, 
then it remains constant upon losing the current item. Intuitively, we could assume that buyer~$i$ did not demand the item 
for sale at $\x$, so we could assume that 
the demand was a subset of the supply at $\x+\e_{-i}$.
 If instead buyer~$i$ wins an item, then its greedy demand can decrease by at most one. 
In particular, 
if buyer~$i$ demands the entire supply at decision node $\x$, upon winning an item, it will continue to demand the entire supply.

\begin{lemma}\label{lem:evolutionMinmaxDemand}
For any $\x \in \mathbb{H}_+$ and any $i \in \{1,2\}$:
\begin{align*}
\kappa_i(\x) < t &\Rightarrow \kappa_i(\x+\e_{-i}) = \kappa_i(\x) \\
t > 1 &\Rightarrow \kappa_i(\x+\e_i) \geq \kappa_i(\x) - 1 \\
t > 1, \kappa_i(\x) = t & \Rightarrow \kappa_i(\x+\e_i) = t-1
\end{align*}
\end{lemma}

As before, we divide Lemma \ref{lem:evolutionMinmaxDemand} into several separate statements. 
First, note that proof of Claim \ref{cl:sdf<t} implies that, if the buyer does not demand the entire supply, then its demand will not change if it loses: 
\begin{corollary}\label{lem:demand<supply}
If the greedy demand of a buyer is less than the entire supply then it remains constant upon losing the current item.
That is,
\begin{equation*}
\kappa_i(\x) < t  \ \Longrightarrow \ \kappa_i(\x+\e_{-i}) = \kappa_i(\x)
\end{equation*}
\end{corollary}

\begin{proof}
By the proof of Claim \ref{cl:sdf<t}, $\mu_i(\x+\e_{-i}) = \max_{k \in [t-1] \cup \{0\}} \bar{\mu}_i(k|\x)$. Note that:
\begin{align*}
\kappa_i(\x+\e_{-i}) &= \min \, \arg\max_{k \in [t-1] \cup \{0\}} \bar{\mu}_i(k|\x) \\
\kappa_i(\x)&=  \min\, \arg\max_{k \in [t] \cup \{0\}} \bar{\mu}_i(k|\x) 
\end{align*}
Suppose that $\kappa_i(\x) \neq \kappa_i(\x+\e_{-i})$ -- this can only be the case if:
$$t = \arg \max_{k \in [t] \cup \{0\}} \bar{\mu}_i(k|\x)$$
This contradicts the fact that $\kappa_i(\x) < t$.
\end{proof}

The next lemma concerns the behaviour of the greedy demand $\kappa_i(\x)$ if buyer $i$ wins: upon purchasing an item, buyer $i$'s greedy demand cannot decrease too much.

\begin{lemma}\label{lem:kappa-minus-1}
The greedy demand of a buyer can decrease by at most one when the buyer wins.
Specifically, for any decision node $\x$ with $t>1$ and any buyer~$i$,
$$\kappa_i(\x+\e_i) \ \geq\  \kappa_i(\x) - 1$$
\end{lemma}
\begin{proof}
If $\kappa_i(\x) \leq 1$, the result is trivial.
So assume that $\kappa_i(\x) > 1$. 
Now, to simplify the notation that follows, set $\kappa_i=\kappa_i(\x)$.
Because $\kappa_i$ is the smallest number of items the buyer can target to obtain its maximum utility via
the greedy bidding strategy, we have $\bar{\mu}_i(k|\x) < \bar{\mu}_i(\kappa_i|\x)$ 
for any $k < \kappa_i$.
This fact gives
\begin{align}\label{eq:kappa-minus-1}
\bar{\mu}_i(\kappa_i-1|\x+\e_i) 
& = \bar{\mu}_i(\kappa_i|\x) - v_1(1|\x) + v_{-i}(t-\kappa_i-1|\x+\e_i) \nonumber \\ 
& > \bar{\mu}_i(k|\x) - v_1(1|\x) + v_{-i}(t-\kappa_i-1|\x+\e_i)\nonumber  \\
& \geq \bar{\mu}_i(k|\x) - v_1(1|\x) + v_{-i}(t-k-1|\x+\e_i) \nonumber \\
& = \bar{\mu}_i(k-1|\x+\e_i)
\end{align}
Here the two equalities are applications of Claim~\ref{cl:win-utility-k}.
The weak inequality follows because the incremental values of buyer~$-i$ are weakly decreasing.

So (\ref{eq:kappa-minus-1}) tells us that $\bar{\mu}_i(\kappa_i-1|\x+\e_i) > \bar{\mu}_i(k-1|\x+\e_i)$
for any $k < \kappa_i$. The result immediately follows.
\end{proof}

Lemma~\ref{lem:kappa-minus-1} implies that if a buyer's short-sighted optimization problem compels it 
to demand the entire supply,  then it will still demand all the remaining supply upon winning that round.
\begin{corollary}\label{cor:entire-supply}
If the greedy demand of a buyer is the entire supply then, upon winning the current item, it will 
continue to demand the entire remaining supply. That is,
\[
\pushQED{\qed}  
\kappa_i(\x) = t \ \Longrightarrow \ \kappa_i(\x+\e_i) = t-1 \qedhere \popQED 
\]
\end{corollary}

\subsubsection{Evolution of Baseline and Threshold Prices}

Finally, we inspect the evolution of baseline and threshold prices. The first lemma shows that the baseline price is a lower bound for the 
threshold price.
Intuitively, if buyer~$i$ can win an item at a price no greater than the baseline price, by Lemma \ref{lem:evolutionMinmaxUtility} 
she would be able to obtain at least her greedy outcome, so her threshold price cannot be less than her baseline price.

\begin{lemma}\label{lem:base-thresh}
For any $\x \in \mathbb{H}_+$ and any $i \in \{1,2\}$, $p_i(\x) \geq \beta_i(\x)$.
\end{lemma}

\begin{proof}
First consider the case $f_i(\x) = 0$. Then buyer~$i$ has no duopsony power and, by Claim~\ref{cl:sdf-0}, we have 
$f_i(\x+\e_j) = 0$, for any buyer $j$. Consequently  $\mu_i(\x+\e_j) = 0$.
Thus, by definition~(\ref{def:threshold-price}), the threshold price satisfies
\begin{align*}
p_{i}(\x) &= v_i (1|\x) + \mu_i(\x+\e_i) - \mu_i(\x+\e_{-i}) \\
&= v_i (1|\x)  
\end{align*}
Similarly, by definition~(\ref{def:baseline-price}), we have $\beta_i(\x) = v_{i}(1|\x)$. Thus the baseline and
threshold prices are equal.

Next, consider the case $f_i(\x) = 1$.
Then $f_i(\x+\e_i) = f_i(\x)-1 =0$ by Claim~\ref{cl:sdf<t},
hence: 
\begin{equation}\label{eq:mu-zero}
\mu_i(\x+\e_i) \ =\ 0
\end{equation}
Now partition the analysis in two further cases. First, assume $t > 1$. Then, again by Claim~\ref{cl:sdf<t}, we have 
$f_i(\x+\e_{-i}) = f_i(\x)= 1$. So the greedy demands satisfy $\kappa_i(\x+\e_{-i})=\kappa_i(\x)=1$. Therefore,
\begin{align}\label{eq:mu-v}
\mu_i(\x+\e_{-i}) &= \bar{\mu}_i(1|\x+\e_{-i}) \nonumber \\
&= v_i(1|\x+\e_{-i}) - v_{-i}(t-1|\x+\e_{-i}) \nonumber \\
& = v_i(1|\x) - v_{-i}(t|\x) \\
& = \bar{\mu}_i(1|\x)  \nonumber \\
&= \mu_i(\x) \nonumber
\end{align}
Because $f_i(\x) > 0$, by definition~(\ref{def:baseline-price}), the baseline price is
\begin{equation}\label{eq:beta-v}
\beta_i(\x)\ =\  v_{-i}(t-\kappa_i(\x)+1|\x) \ =\ v_{-i}(t|\x)
\end{equation}
Thus, the threshold price is given by
\begin{align*}
p_{i}(\x) &= v_i (1|\x) + \mu_i(\x+\e_i) - \mu_i(\x+\e_{-i}) \\
&= v_i (1|\x) - \mu_i(\x+\e_{-i}) \\
&= v_i (1|\x)  -  \left(v_i(1|\x) - v_{-i}(t|\x)   \right)\\
&=   v_{-i}(t|\x)  \\
&=  \beta_i(\x) 
\end{align*}
Here the second, third and fifth equalities follows from (\ref{eq:mu-zero}), (\ref{eq:mu-v}) and (\ref{eq:beta-v}), respectively.
Thus, again, the baseline and threshold prices are equal.

Second assume instead that $t = 1$. Because this is then the final round, we have 
$\mu_i(\x+\e_i) = \mu_i(\x+\e_{-i}) = 0$.
The threshold price is then $p_i(\x) = v_i(1|\x)$. 
Now as $f_i(\x) = 1$ we have $v_i(1|\x) > v_{-i}(t|\x) = \beta_i(\x)$.
So the  threshold price is strictly greater than the baseline price.

Finally consider the case $f_i(\x) > 1$. 
Then, by Lemma~\ref{lem:sdf}, we have $f_i(\x+\e_j) \geq f_i(\x) -1 \ge 1$ for any buyer $j$. 
Now, denoting $\kappa_i(\x)$ as $\kappa_i$,
\begin{align}\label{eq:mu-into-beta}
\mu_i(\x+\e_i)
& = \max_{k \in [t-1]} \bar{\mu}_i(k|\x+\e_{i}) \nonumber \\
& \geq \bar{\mu}_i(\kappa_i-1|\x+\e_i) \nonumber \\
& = \bar{\mu}_i(\kappa_i|\x) - v_i(1|\x) + v_{-i}(t-\kappa_i+1|\x) \nonumber \\
& = \bar{\mu}_i(\kappa_i|\x) - v_i(1|\x) + \beta_i(\x)
\end{align}
Here the inequality follows from considering a specific solution to the maximization problem. The first 
equality follows from Claim \ref{cl:win-utility-k}, and the final equality follows from the definition of~(\ref{def:baseline-price}) of the 
baseline price $\beta_i(\x)$. 
Plugging (\ref{eq:mu-into-beta}) into the definition~(\ref{def:threshold-price}) of the threshold price, we obtain:
\begin{align}\label{eqn:last-case}
p_i(\x)  &= v_i(1|\x) + \mu_i(\x+\e_i) - \mu_i(\x+\e_{-i}) \nonumber\\
& \geq   v_i(1|\x)  + \bar{\mu}_i(\kappa_i|\x) - v_i(1|\x) + \beta_i(\x) - \mu_i(\x+\e_{-i}) \nonumber\\
& = \beta_i(\x) + \mu_i(\x)  - \mu_i(\x+\e_{-i})
\end{align}
By Lemma~\ref{lem:evolutionMinmaxUtility}, we know that $\mu_i(\x) \ge \mu_i(\x+\e_{-i})$.
Hence~(\ref{eqn:last-case}) implies that the threshold price is at least the baseline price.
\end{proof}

Moreover, if buyer~$i$'s greedy demand corresponds to the entire supply, it should ensure that it wins every item while targeting his greedy utility. 
This implies the inequality of Lemma \ref{lem:base-thresh} should become strict.

\begin{lemma}\label{lem:strictly-greater}
If the greedy demand of a buyer is the entire supply then the threshold price is strictly greater than the baseline price.
Specifically,
\begin{equation*}
\kappa_i(\x)=t \ \Longrightarrow\  p_i(\x) >  \beta_i(\x)
\end{equation*}
\end{lemma}

\begin{proof}
By assumption $\kappa_i(\x)=t$,
hence by Corollary~\ref{cor:entire-supply}, we have $\kappa_i(\x+\e_i)=t-1$. Then:
\begin{equation}\label{eq:A}
\mu_i(\x+\e_i) \ =\ \bar\mu_i(t-1 | \x+\e_i) 
\end{equation}
Furthermore, by Claim~\ref{cl:win-utility-k},
\begin{equation}\label{eq:B}
\bar{\mu}_i(t-1|\x+\e_i) \ =\  \bar{\mu}_i(t|\x) + v_i(1|\x) - v_{-i}(1|\x)
\end{equation}
Finally, by Lemma~\ref{lem:evolutionMinmaxUtility}, we have 
\begin{equation}\label{eq:C}
\mu_i(\x) > \mu_i(\x+\e_{-i}) 
\end{equation}
Thus the threshold price at the decision node $(\x,t)$ satisfies
\begin{align}\label{eq:D}
p_i(\x) &= v_i(1|\x) + \mu_i(\x+\e_i) - \mu_i(\x+\e_{-i}) \nonumber \\
&= v_i(1|\x) + \bar\mu_i(t-1 | \x+\e_i) - \mu_i(\x+\e_{-i}) \nonumber \\
&= v_i(1|\x) + \left( \bar{\mu}_i(t|\x) - v_i(1|\x) + v_{-i}(1|\x)\right) - \mu_i(\x+\e_{-i}) \nonumber  \\
&=  v_{-i}(1|\x) + \bar\mu_i(t|\x)  - \mu_i(\x+\e_{-i})\nonumber  \\
&= v_{-i}(1|\x) + \mu_i(\x)  - \mu_i(\x+\e_{-i}) \nonumber \\
&> v_{-i}(1|\x) 
\end{align}
Here the second and third equalities are due to~(\ref{eq:A}) and~(\ref{eq:B}), respectively. 
The strict inequality follows from~(\ref{eq:C}). 

To conclude, recall that $\kappa_i(\x)=t$ implies that the baseline price is 
$\beta_i(\x)=v_{-i}(t-\kappa_i(\x)+1|\x)=v_{-i}(1|\x)$. Substituting into~(\ref{eq:D})
gives $p_i(\x)>\beta_i(\x)$ and so the threshold price strictly exceeds the baseline price.
\end{proof}

Instead consider the case when some buyer~$i$ with duopsony power does not demand the entire supply. 
Suppose that, buyer~$i$ wins an item, and still has duopsony power after doing so. 
As its demand will not decrease significantly, its baseline price will be weakly higher. We would then presume that 
buyer~$i$ is in a situation that favours buying more items, hence it would be willing to pay higher prices.

\begin{lemma}\label{lem:evolutionThresholdWin}
Given $\x \in \mathbb{H}_+$ and $i \in \{1,2\}$ such that $f_i(\x) > 1$ and $\kappa_i(\x) < t$. Then $\beta_i(\x+\e_i) \geq p_i(\x)$,
where equality holds only if $\bar{\mu}(\kappa_i(\x+\e_i)+1|\x) = \mu_i(\x)$. Moreover, $p_i(\x+\e_i) \geq p_i(\x)$.
\end{lemma}

\begin{proof}
Let us first  prove the first part of the statement. As $\kappa_i(\x) < t$,  we have:
$\mu_i(\x)= \mu_i(\x+\e_{-i})$ by Lemma~\ref{lem:evolutionMinmaxDemand}. 
Setting $\hat{k}=\kappa_i(\x+\e_i)$, the threshold price is then
\begin{align}\label{eq:pee1}
p_i(\x) 
& = v_i(1|\x) + \mu_i(\x+\e_i) - \mu_i(\x+\e_{-i}) \nonumber \\
& = v_i(1|\x) + \mu_i(\x+\e_i) - \mu_i(\x) \nonumber \\
& = v_i(1|\x) + \sum_{j=1}^{\hat{k}} v_i(j|\x+\e_i)  - \hat{k} \cdot v_{-i}(t-\hat{k}|\x+\e_i) - \mu_i(\x) \nonumber \\
& = v_i(1|\x) + \sum_{j=1}^{\hat{k}} v_i(j+1|\x)  - \hat{k} \cdot v_{-i}(t-\hat{k}|\x+\e_i) - \mu_i(\x)  \nonumber \\
& = \sum_{j=1}^{\hat{k}+1} v_i(j|\x)  - \hat{k} \cdot v_{-i}(t-\hat{k}|\x+\e_i) - \mu_i(\x)  
\end{align}
On the other hand,
\begin{align}\label{eq:pee2}
\hat{k} \cdot v_{-i}(t-\hat{k}|\x+\e_i) &=   (\hat{k}+1) \cdot v_{-i}(t-\hat{k}|\x+\e_i) - v_{-i}(t-\hat{k}|\x+\e_i) \nonumber\\
&=   (\hat{k}+1) \cdot v_{-i}(t-\hat{k}|\x)  - v_{-i}(t-\hat{k}|\x+\e_i)\nonumber \\
&=   (\hat{k}+1) \cdot v_{-i}(t-\hat{k}|\x)  - \beta_{i}(\x+\e_i) 
\end{align}
Together (\ref{eq:pee1}) and (\ref{eq:pee2}) give 
\begin{align*}
p_i(\x) & = \sum_{j=1}^{\hat{k}+1} v_i(j|\x)  - \left( (\hat{k}+1) \cdot v_{-i}(t-\hat{k}|\x)  - \beta_{i}(\x+\e_i)  \right)  - \mu_i(\x)  \\
& = \bar{\mu}_i(\hat{k}+1|\x) + \beta_{i}(\x+\e_i) - \mu_i(\x) \\
& \leq \beta_i(\x+\e_i)
\end{align*}
Here, the inequality is tight only if $\mu_i(\x) = \bar{\mu}_i(\hat{k}+1|\x)$.

Now we show that the threshold price weakly increases -- we have two cases to analyse. First, assume that $\kappa_i(\x+\e_i) < t-1$. 
Then:
\begin{align*}
p_i(\x+\e_i) &\geq \beta_i(\x+\e_i) \\
&\geq p_i(\x)
 \end{align*}
Here the first inequality follows from Lemma~\ref{lem:base-thresh}.

Second, assume that $\kappa_i(\x+\e_i) = t-1$. Then:
\begin{align*}
p_i(\x+\e_i) &> \beta_i(\x+\e_i) \\
&\geq p_i(\x)
 \end{align*}
Now the first inequality instead follows from Lemma~\ref{lem:strictly-greater}.
\end{proof}

If instead buyer $-i$ wins at $\x$, then buyer $i$ loses the opportunity to apply its greedy strategy to win $t$ items from $\x$. 
If buyer $i$ still does not demand the entire supply at $\x+\e_{-i}$, then this loss of opportunity translates to a lesser incentive to purchase 
at a given price:

\begin{lemma}\label{lem:threshold-decreasing}
Given $\x \in \mathbb{H}_+$ and $i \in \{1,2\}$ such that $f_i(\x) > 1$ but $\kappa_i(\x) < t-1$. Then $p_i(\x+\e_{-i},t-1) \ \leq\ p_i(\x,t)$. Moreover, 
the inequality is strict if and only if $\kappa_i(\x+\e_i) = t-1$.
\end{lemma}

\begin{proof}
Take $\kappa_i(\x) < t-1$. Then, by definition~(\ref{def:threshold-price}) of threshold prices, we have that:
\begin{align*}
p_i(\x) - p_i(\x+\e_{-i}) &= \text{\hspace{2mm}} \big( v_i(1|\x) + \mu_i(\x+\e_i) - \mu_i(\x+\e_{-i}) \big)\\ 
&\ \ \ \ \ - \big( v_i(1|\x+\e_{-i}) + \mu_i(\x+\e_i+\e_{-i}) - \mu_i(\x+2\e_{-i}) \big)\\
&= \text{\hspace{2mm}}\mu_i(\x+\e_i) - \mu_i(\x+\e_i+\e_{-i}) \\
&\geq \text{\hspace{2mm}}0
\end{align*}
Here the first equality follows by the definition of threshold prices. The second equality and the inequality follow from 
Lemma \ref{lem:evolutionMinmaxUtility}. Again by Lemma \ref{lem:evolutionMinmaxUtility}, the inequality is 
strict if and only if $\kappa_i(\x+\e_i) = t-1$.
\end{proof}

Finally, if buyer~$i$ with duopsony power targets his greedy utility and does not demand the entire supply, 
then incentives for buyer~$-i$ are aligned such 
that it should want to purchase items without letting buyer~$i$ win. Buyer~$-i$ would be able to do so if buyer~$i$'s bids 
do not exceed buyer~$-i$'s incremental value. The 
following lemma shows that this is the case.

\begin{lemma}\label{cl:top-marginal}
Given $\x \in \mathbb{H}_+$ and $i \in \{1,2\}$. If $\kappa_i(\x) < t$ and $f_i(\x) > 0$  then $p_i(\x) \ \leq\  v_{-i}(t-\kappa_i(\x+\e_i)|\x)$.
Moreover, the inequality 
is tight if only if $\bar{\mu}_i(\kappa_i(\x+\e_i)+1|\x) = \mu_i(\x)$.
\end{lemma}

\begin{proof}
Set $\hat{k}=\kappa_i(\x+\e_i)$. Then, by definition~(\ref{def:threshold-price}) of threshold prices,
\begin{align}\label{eq:top}
p_i(\x) & = v_i(1|\x) + \mu_i(\x+\e_i) - \mu_i(\x+\e_{-i}) \nonumber \\
& = v_i(1|\x) + \mu_i(\x+\e_i) - \mu_i(\x) \nonumber \\
& \leq v_i(1|\x) + \mu_i(\x+\e_i) - \bar{\mu}_i(\hat{k}+1|\x) \nonumber \\
& = v_i(1|\x) + \mu_i(\x+\e_i) - \left( \bar{\mu}_i(\hat{k}|\x+\e_i) + v_i(1|\x) - v_{-i}(t-\hat{k}|\x) \right) \nonumber \\
& = v_{-i}(t-\hat{k}|\x) + \mu_i(\x+\e_i) - \bar{\mu}_i(\hat{k}|\x+\e_i) \nonumber \\
& = v_{-i}(t-\hat{k}|\x)
\end{align}
The second equality holds by Lemma~\ref{lem:evolutionMinmaxUtility} because $\kappa_i(\x) < t$.
The first inequality holds due to the optimality of the choice $\kappa_i(\x)$.
The third equality is an application of Claim~\ref{cl:win-utility-k} with $k= \hat{k}+1=\kappa_i(\x+\e_i)+1$.
The fifth equality holds as $\hat{k}=\kappa_i(\x+\e_i)$.
For (\ref{eq:top}) to be tight, we require that $\bar{\mu}_i(\hat{k}+1|\x) = \mu_i(\x)$.
\end{proof}

\subsubsection{Auction Outcomes with Greedy Bidding Strategies}

We are now ready to use our analysis of quantities induced by greedy bidding to inspect what happens in the auction when 
both buyers use their greedy bidding strategies. Note that \emph{realised} quantities are those reached in the outcome with 
positive probability. Tie-breaking rules can introduce randomisation of the allocations and prices buyers face each round, but 
we are able to make a statement for any realisation of these quantities.

\begin{theorem}\label{thm:greedy-equilibria}
Suppose buyers implement their greedy bidding strategies. Then on any realised outcome path from some decision node $\x$, if 
there exists a monopsonist buyer $i$ at $\x$, then her realised utility is $\mu_i(\x)$; else some 
buyer $i \in \arg\hspace{-.05cm}\min_{j \in \{1,2\}} p_i(\x)$ has realised utility equal to $\mu_i(\x)$. Furthermore, buyer $i$ 
purchases at least $\kappa_i(\x)$ items. Finally, prices are equal to $p_i$ alongside the realised outcome path until 
buyer $i$ demands the entire supply, after which prices equal $\beta_i$. In particular, prices are declining 
along any realised outcome path.
\end{theorem}

\begin{proof}
Proceed by induction on $t$: if $t = 1$, then $p_i(\x) = v_i(1|\x)$ for any buyer $i$, so we have the standard equilibrium for a single-item second-price auction, which satisfies the necessary conditions. 

So suppose that $t > 1$. First suppose that no buyer has any duopsony power. Then by definition \ref{def:duopsony-factor} of duopsony power, for any $i,j \in \{1,2\}$ and $k,l \in [t]$, $v_i(k) = v_j(l)$, i.e. all incremental valuations are equal to some $w \geq 0$. In this case, at every node $\x'$ of the auction rooted at decision node $\x$, $p_i(\x') = \beta_i(\x') = w$, and $u_i(\x') = \mu_i(\x') = 0$. This satisfies the conditions of the theorem.

Now suppose that buyer $i$ is a monopsonist at $\x$ with positive duopsony power, then by Lemma \ref{lem:sdf}, buyer $i$ is a monopsonist at every node of the auction tree rooted at $\x$. Therefore, buyer $i$ attains its greedy utility in every subsequent node. Now, by Lemma \ref{lem:strictly-greater} if buyer $i$ demands the entire supply, then $\min \{ p_i(\x), v_i(1|\x) \} > \beta_i(\x) = p_{-i}(\x)$, so buyer $i$ wins with probability $1$ at $\x$. At decision node $\x+\e_i$, by Corollary \ref{cor:entire-supply}, buyer $i$ demands the entire supply, so by the induction hypothesis she keeps winning at price $v_{-i}(1|\x) = \beta_i(\x)$. Then prices are constant along the outcome path, and the realised utility of buyer $1$ is $\sum_{k = 1}^{T }v_i(k|\x) - t \cdot v_{-i}(1|\x) = \mu_i(\x)$.

Instead suppose $\kappa_i(\x) < t$. Then $p_{-i}(\x) = v_{-i}(1|\x)$, and by Lemma \ref{cl:top-marginal}, $p_i(\x) \leq v_{-i}(1|\x)$. Therefore, the price equals $p_i(\x)$ at decision node $\x$. If buyer $-i$ wins at decision node $\x$, then buyer $i$'s realised utility equals $\mu_i(\x+\e_{-i})$, which by Lemma \ref{lem:evolutionMinmaxUtility} equals $\mu_i(\x)$. Furthermore, by Lemma \ref{lem:threshold-decreasing}, $p_i(\x+\e_{-i}) \leq p_i(\x)$, so prices are non-increasing. Now suppose that buyer $i$ wins at decision node $\x$. Then $p_i(\x) = v_{-i}(1|\x)$. By Lemma \ref{cl:top-marginal}, this can only happen if $v_{-i}(1|\x) = v_{-i}(t-\kappa_i(\x+\e_i)|\x) = \beta_i(\x+\e_{-i})$ and $\mu_i(\x+\e_i) + v_i(1|\x) - v_{-i}(1|\x) = \mu_i(\x)$. Then prices equal to $v_{-i}(1|\x)$ for the rest of the auction, and buyer $i$ earns her greedy utility.

Finally, suppose $f_i(\x) > 0$ for any $j \in \{1,2\}$, and that buyer $i$ wins at $\x$. Then price is equal to $p_{-i}(\x)$. If buyer $-i$ is a monopsonist at $\x+\e_i$, then by the induction hypothesis and Lemma \ref{lem:threshold-decreasing}, we have the desired result. If not, then by Lemma \ref{lem:evolutionThresholdWin} and Lemma \ref{lem:threshold-decreasing}, $p_i(\x+\e_i) \geq p_i(\x) \geq p_{-i}(\x) \geq p_{-i}(\x+\e_i)$. If buyer $-i$ is the member of $\arg\hspace{-.05cm}\min_{j \in \{1,2\}} p_i(\x+\e_i)$ who has realised utility equal to $\mu_{-i}(\x+\e_i)$, then we have the desired result. If instead it is buyer $i$, note that $p_i(\x+\e_i) = p_{-i}(\x+\e_i)$. Therefore, the price at $\x$ equals $p_i(\x)$, and prices do not increase from $\x$ to $\x+\e_i$. Finally, as $\mu_i(\x+\e_{-i}) = \mu_i(\x+\e_i)$ by Lemma \ref{lem:evolutionMinmaxUtility}, $\mu_i(\x) = v_i(1|\x) - p_i(\x) + \mu_i(\x+\e_i)$, hence buyer $i$'s realised utility equals her greedy utility.

It remains to show that each buyer $j$ purchases at least $\kappa_j(\x)$ items from any outcome path starting from $\x$. The proof is again by induction on $t$ -- if $t = 1$, a buyer $j$ with $\kappa_j(\x) = 1$ has $v_j(1|\x) > v_{-j}(1|\x)$, hence he necessarily wins an item. Now suppose that $t > 1$. If $\kappa_j(\x) = t$, then as argued above buyer $j$ purchases all items at price $\beta_j(\x)$, hence buyer $j$ indeed purchases $t$ items. If $\kappa_j(\x) < t$, then $\kappa_j(\x+\e_j) \geq \kappa_j(\x) - 1$ and $\kappa_j(\x+\e_{-j}) = \kappa_j(\x)$, so whoever wins at $\x$, a realised outcome path proceeds to an outcome which awards buyer $j$ a number of items $\geq \kappa_j(\x)$ onwards from $\x$.
\end{proof}
\section{Characterisation of No-Overbidding Equilibria}\label{sec:equilibria}

In this section, we classify the equilibria of two-buyer sequential multiunit auctions under the no-overbidding condition. 
Structurally, we will see that any equilibrium is made up of three phases (a competitive phase, a competition reduction phase 
and a monopsony phase) characterized by very different strategic behaviours. 

First, however, let's show that the \emph{declining price anomaly} holds. Here, it is worth emphasizing declining prices do not follow as a direct 
consequence of the no-overbidding assumption. Indeed, for $\geq 3$ buyers, the declining price anomaly can fail to hold given
decreasing incremental valuations even with the no-overbidding assumption; see Prebet et al.~\cite{PNV19}. 
The proof of the declining price anomaly depends instead on a weaker version of Equation $(8)$ in \cite{GS01}, which survives the imposition of the no-overbidding constraint.

\begin{lemma}\label{lem:priceBound}
Let $\x$ be a decision node and let $p(\x)$ denote the price paid at $\x$. Then for any buyer $i$, $p(\x) \leq v_i(1|\x) + u_i(\x+\e_i) - u_i(\x)$.
\end{lemma}

\begin{proof}
If buyer $i$ wins with probability $1$ at decision node $\x$, then the inequality holds with equality. Otherwise, suppose that buyer $i$ wins with probability $q < 1$ at decision node $\x$. Then:
$$u_i(\x) = q \cdot \left(v_i(1|\x) - p(\x) +  u_i(\x+\e_i)\right) + (1-q) \cdot u_i(\x+\e_{-i})$$
Now suppose for a contradiction that $p(\x) > v_i(1|\x) + u_i(\x+\e_i) - u_i(\x)$. Then, plugging in the above expression:
$$p(\x) > v_i(1|\x) + u_i(\x+\e_i) - u_i(\x+\e_{-i})$$
However, $p(\x) = \min_{j \in \{1,2\}} b_j(\x) \leq b_i(\x)$, and by Equation \ref{eq:bid-value}:
$$b_i(\x) \leq v_i(1|\x) + u_i(\x+\e_i) - u_i(\x+\e_{-i}) \Rightarrow\!\Leftarrow$$
\end{proof}

The lemma allows us to adapt the declining price anomaly proof of Gale and Stegeman~\cite{GS01}:

\begin{theorem}\label{thm:dpa}
In a two-buyer sequential multiunit auction, under equilibrium bidding strategies, prices are non-increasing 
along any realised equilibrium path.
\end{theorem}

\begin{proof}
Suppose that WLOG buyer $1$ wins with positive probability at decision node $\x$, with $t > 1$ -- we will show that $p(\x) \geq p(\x+\e_1)$. If buyer $2$'s no-overbidding constraint was binding at $\x$, then we are done, since $p(\x+\e_1) \leq v_2(1 |\x) = p(\x)$. So suppose that buyer $2$'s no-overbidding constraint does not bind, then:
$$p(\x) = v_2(1|\x) + u_2(\x+\e_2) - u_2(\x+\e_1)$$
On the other hand, by Lemma \ref{lem:priceBound}:
$$p(\x+\e_1) \leq v_2(1|\x) + u_2(\x+(1,1)) - u_2(\x+\e_1)$$
Subtracting the two equations, we get:
$$p(\x) - p(\x+\e_1) \geq u_2(\x+\e_2) - u_2(\x+(1,1))$$
However, the RHS must be $\geq 0$, or buyer $2$ would have a profitable deviation at $\x+\e_2$ by bidding $b_1(\x) - \epsilon$ for $\epsilon > 0$. That is impossible since bidding strategies are the unique ones surviving the iterative elimination of weakly dominated strategies.
\end{proof}

We now proceed to show when there necessarily is a direct equivalence between greedy bidding and equilibrium bidding 
strategies: it is exactly when there exists a monopsonist. 

\begin{theorem}\label{thm:bid-threshold-main}
Suppose that at decision node $\x$, some buyer $i$ is a monopsonist. Then for any decision node $\x'$ of the 
auction tree rooted at $\x$, prices and utilities are equal for equilibrium and greedy bidding strategies.
\end{theorem}

\begin{proof}
Proceed by induction on $t$: for $t = 1$, the statement is obvious, so suppose $t > 1$, and consider equilibrium bidding strategies from decision node $\x$, where buyer $i$ is a monopsonist. As the existence of a monopsonist is a hereditary property on the auction graph, buyer $i$ is again a monopsonist at the successor nodes to $\x$. Therefore, $u_i(\x+\e_j) = \mu_i(\x+\e_j)$ for any decision node, hence buyer $i$ bids $b_i(\x) = \min\{ p_i(\x), v_1(1|\x) \}$ as intended.

First consider the case if buyer $i$ has no duopsony power, then both buyers are monopsonists. As being a monopsonist is a hereditary property, $u_j(\x+\e_\ell) = 0$ for any $j,\ell \in \{1,2\}$ by the induction hypothesis, hence each buyer $j$ bids $p_j(\x) = v_j(1|\x)$. In this case, bids for greedy and equilibrium bidding strategies are equal at $\x$, hence by the induction hypothesis prices and utilities are equal for equilibrium and greedy bidding strategies as desired.

So we consider the case if buyer $i$ does have duopsony power. If $p_i(\x) > v_{-i}(1|\x)$, then by Lemma \ref{lem:strictly-greater} and Lemma \ref{cl:top-marginal}, buyer $i$ demands $t$ items. Then buyer $i$ purchases an item at price $\beta_i(\x) = v_{-i}(1|\x)$ at decision node $\x$. Then by Lemma \ref{lem:evolutionMinmaxDemand}, $\kappa_i(\x+\e_1) = t-1$, so by the induction hypothesis, buyer $i$ keeps purchasing the rest of the items at price $v_{-i}(1|\x)$ as desired.

Now suppose that $p_i(\x) \leq v_{-i}(1|\x)$, then by Lemma \ref{lem:strictly-greater}, buyer $i$ demands $< t$ items. Furthermore, as buyer $i$ has duopsony power then buyer $i$ bids $b_i(\x) = p_i(\x)$ by Lemma~\ref{cl:top-marginal}. First consider the case if buyer $-i$'s no-overbidding constraint binds -- then buyer $-i$ bids $v_{-i}(1|\x)$, so bids at decision node $\x$ are equal for equilibrium and greedy bidding strategies. Then by the induction hypothesis the outcomes from decision node $\x$ are also equal for equilibrium and greedy bidding strategies. Therefore, we need to only consider the case if $b_{-i}(x) = v_{-i}(1|\x) + u_i(\x+\e_{-i}) - u_i(\x+\e_i) < v_{-i}(1|\x)$. Note that in this case, buyer $-i$ must have positive forward utility at $\x+\e_{i}$, hence buyer $i$ does not demand the entire supply after winning an item.

If buyer $-i$ wins at decision node $\x$, then the outcome at decision node $\x$ does not change if buyer $-i$ deviates to bidding $v_{-i}(1|\x)$, so suppose that buyer $i$ wins at decision node $\x$ under equilibrium bidding strategies. By the induction hypothesis, the price at decision node $\x+\e_i$ equals $p_i(\x+\e_i)$ as buyer $i$ does not demand the entire supply at $\x+\e_i$. However, by Lemma~\ref{lem:evolutionThresholdWin} $p_i(\x+\e_i) \geq p_i(\x)$. On the other hand, as buyer $i$ wins at decision node $\x$ under equilibrium bidding strategies, by Theorem \ref{thm:dpa} $p_i(\x+\e_i) \leq p_i(\x)$. Therefore $p_i(\x+\e_i) = p_i(\x)$. So buyer $i$ must have won at price no less than $p_i(\x)$. As buyer $i$ bid $p_i(\x)$ at $\x$, we conclude that buyer $i$ pays price $p_i(\x)$ at $\x$. By Lemma \label{cl:top-marginal} and as buyer $i$ has duopsony power, $b_i(\x) = p_i(\x) < v_i(1|\x)$, and by assumption $b_{-i}(x) < v_{-i}(1|\x)$. In particular, either buyer is indifferent to winning and losing at price $p_i(\x)$. This implies that buyer $-i$'s deviation to $v_{-i}(1|\x)$ to win the item at $\x$ does not change the price and utilities at decision node $\x$.
\end{proof}

Informally, suppose only buyer $1$ has duopsony power; then buyer $1$ is constrained by the equilibrium
bidding strategies to make her greedy utility at every possible 
future node. Thus her bids equal to her threshold price at every round of the auction. Buyer $2$ will then take 
advantage of buyer $1$'s bidding strategies by purchasing an item 
whenever possible.


But what happens in the more complex setting where both buyers have duopsony power? Call buyer $i$ 
a \emph{quasi-monopsonist} at decision node $\x$ if there exists a realised equilibrium path 
from $\x$ to a final round $\y$ such that $b_i(\y) \geq b_{-i}(\y)$.  Note that it is possible for both agents to 
be quasi-monopsonists at a node $\x$ if there is a randomized tie-breaking rule.  
By decreasing prices, a quasi-monopsonist 
may have a realised payoff weakly less than its greedy utility; moreover, a monopsonist is always a quasi-monopsonist. 
This definition, along with properties of greedy bidding, allows us to fully characterise equilibria.


\begin{theorem}\label{thm:eql-character}
For equilibrium bidding strategies, while no buyer demands the entire supply, prices at each node are no less than 
the minimum threshold price. Moreover, at every decision node $\x$ there exists a quasi-monopsonist buyer $i$. 
Finally, if buyer $i$ is a quasi-monopsonist at decision node $\x$ and if $\x+\e_{-i}-\e_i$ is also a decision node, 
then $i$ is again a quasi-monopsonist at decision node $\x+\e_{-i}-\e_i$.
\end{theorem}

\begin{proof}
Begin by showing the first statement of the theorem: let $\x$ be a decision node. If there exists a monopsonist buyer $i$ at decision node $\x$, then prices at each node equal $p_i$ as long as buyer $i$ does not demand the entire supply, so suppose that both buyers have duopsony power. Consider a realised equilibrium path from decision node $\x$; on this path, there must be an earliest decision node $\x'$ such that some buyer $i$ is a monopsonist. As $\x'$ is the earliest such node, by Lemma \ref{lem:sdf} the previous node on the equilibrium path must be $\y$, where $\x' = \y+\e_{-i}$.

Now, $b_i(\y) \geq p_i(\y)$, as: (1) buyer $i$ attains her greedy utility from decision node $\x'$, and must earn at least her greedy utility in expectation from decision node $\y+\e_i$, and (2) by Lemma \ref{cl:top-marginal} buyer $i$'s no-overbidding constraint would not bind if she were to bid $p_i(\y)$. Furthermore, as buyer $-i$ wins at decision node $\y$, price at $\y$ equals $b_i(\y)$. However, tracing back to decision node $\x$, by Lemma \ref{lem:evolutionThresholdWin} and Lemma \ref{lem:threshold-decreasing}\footnote{Buyer $i$ cannot demand the entire bundle at any node from $\y$ to $\x$, as she is not a monopsonist.}, $p_i$ is decreasing, while prices are non-decreasing. Therefore, at decision node $\x$, the price is no less than $p_i(\x)$, hence is no less than the minimum threshold price.

The existence of a quasi-monopsonist is obvious, so we prove the ``shift property''. Suppose that buyer $i$ is a quasi-monopsonist at decision node $\x$, and let $P$ be a realisation of the equilibrium path such that buyer $i$'s bid at the final round is no less than buyer $-i$'s bid at the final round. Now, take an arbitrary realisation of the equilibrium path $Q$ from $\x+\e_{-i}-\e_{i}$. If $Q$ intersects $P$, then the path starting from $\x+\e_{-i}-\e_{i}$ and following $Q$ until intersecting $P$, and then following $P$ is an equilibrium path. Therefore, $i$ is a quasi-monopsonist at $\x+\e_{-i}-\e_{i}$. 

Suppose instead that $Q$ never intersects $P$. Then if $\x', \x''$ are respectively nodes on $P$ and $Q$ with $t(\x') = t(\x'') = s$, it must be that $\x'' = k(s) \cdot (\e_{-i}-\e_{i}) + \x'$ for some $k(s) > 0$ -- as $k(s) - k(s+1) \in \{-1,0,1\}$ and as the initial nodes have difference $k(t) = 1$. In particular, if $\y$ is a node on $P$ with $t(\y) = 1$, then $\y + k(1) \cdot (\e_{-i}-\e_{i})$ is a node on $Q$. Then at the last round of the auction on $Q$, buyer $i$ bids $v_{i}(1-k(1)|\y) \geq v_{i}(1|\y)$ and buyer $-i$ bids $v_{-i}(1+k(1)|\y) \leq v_{-i}(1|\y)$. By choice of $P$, $v_{i}(1|\y) \geq v_{-i}(1|\y)$. Therefore, buyer $i$ weakly outbids buyer $-i$ at the final round of the auction on realised equilibrium path $Q$, hence buyer $i$ is a quasi-monopsonist at decision node $\x+\e_{-i}-\e_{i}$.
\end{proof}

It may not be immediately apparent, but Theorem~\ref{thm:eql-character} gives us a very clear picture of what happens at 
an equilibrium. Specifically, each equilibrium consists of three phases. The first phase is the {\em competitive phase}.
In this phase the identity of the ``eventual monopsonist'' may change depending on the winner of an item (and may be 
uncertain due to randomized tie-breaking). Consequently, the two buyers compete to buy items and drive 
prices above the threshold prices. The buyer who fails to win enough items in this phase retains sufficient duopsony 
power to become a monopsonist. The second phase, the {\em competition reduction} phase, begins once the identity of 
the monopsonist is established.  The monopsonist then posts its threshold price in each round. The other buyer exploits 
the monopsonist's bidding strategy to purchase items. This phase ends when the competition from the other buyer has been weakened sufficiently 
enough for the monopsonist to desire winning all the remaining items. Thus we enter the third phase, the {\em monopsony phase}, where the monopsonist 
purchases all the remaining items at its current baseline price.

\section{The Price of Anarchy}\label{sec:PoA}

The {\em price of anarchy} of a sequential auction is the worst-case ratio between the social welfare attained at an equilibrium allocation
and the welfare of the optimal allocation. 
In this section, we show first show the flaw with the previously claimed proof of the price of anarchy
for two-buyer sequential auctions \emph{with} overbidding allowed. This motivates us to inspect the price of anarchy under 
the no-overbidding constraint. We then exploit our equilibrium characterization to prove that the price of anarchy is exactly $1-1/e$ in two-buyer
sequential auctions with no-overbidding, assuming weakly decreasing incremental valuations.

\subsection{On the Price of Anarchy of Equilibria with Overbidding}\label{sec:PoA-overbidding}
For two-buyer sequential auctions with overbidding permitted, an identical guarantee of $1-1/e$ on 
the price of anarchy was claimed by Bae et al.~\cite{BBB08,BBB09}. 
The proof in the original paper~\cite{BBB08} was flawed and a new proof was given in~\cite{BBB09}. 
Unfortunately, as we now explain there is also a flaw in the new proof. 

Interestingly, {\sc Example 4} also illustrates the bug in~\cite{BBB09}.
Observe, from Figure~\ref{fig:overbidding}, that the efficiency of this auction is:
$$\frac{1 + (2/3-2\epsilon/3) + (1/2+\epsilon)}{3} \ =\  \frac{13}{18} + \frac{\epsilon}{9}$$
This gives a corresponding price of anarchy of $13/18=0.722$. Of course, this in itself is not problematic 
as it exceeds $1-1/e=0.632$. To explain the problem requires the following notion
defined in~\cite{BBB09}, which we present in the terminology of this paper. 
Assume that the valuations are non-decreasing and concave, and they satisfy 
$v_1(k) = 1$ for any $k \in [T]$ and 
$v_2(1) \leq v_1(T)$, then we say that the marginal utilities satisfy the \textbf{subgame kink property} if:
\begin{enumerate}
\item On the equilibrium path, buyer~$1$ wins all his items last.
\item At any node off the equilibrium path, buyer~$1$ wins.
\end{enumerate}
Notice that our auction in Figure~\ref{fig:overbidding} satisfies the subgame kink property. 
Furthermore, in Lemma~5 of~\cite{BBB09}, it is suggested that:
\begin{lemma*}\cite{BBB09}
For any $T \in \mathbb{N}$, if $v_1(k) = 1$ for any $k \in [T]$ and $b \leq 1$, amongst valuation profiles that support equilibria satisfying the 
subgame kink property, the minimal efficiency valuation profile for buyer~$2$ with the constraint $v_2(1) \leq b$ is 
given by, for some $k \in [T]$:
\begin{equation}\label{eqn:minimalV2}
v_2(j)  = \begin{cases}1 - \frac{k}{T-j+1} & j \in [T-k] \\
 0 & \text{otherwise}\end{cases}
\end{equation}
\end{lemma*}
In particular, for the auction of {\sc Example 4}, we have $v_2(1) \leq 2/3 - \delta$ for small $\delta > 0$.
Consider applying the lemma above for $b = 2/3 - \delta$. Then the minimal efficiency $v_2(\cdot)$ should be given by (\ref{eqn:minimalV2}) 
for some $k \in [T]$. It can't be that $k = 1$, as that would imply $v_2(1) = 2/3 > 2/3 - \delta$. On the other hand, $k = 3$ leads to 
efficiency. So, in this case, the minimal efficiency should be given by $k=2$ where:
\begin{equation*}
v_2(j) = \begin{cases}
1/3 & j = 1 \\
0 & j \in \{2,3\}
\end{cases}
\end{equation*}
Let us compute the resultant equilibrium bidding strategies. In the final round of the auction, it is always the case that buyer~$1$ wins 
the item and pays buyer~$2$'s incremental value. So for $t = 1$, $u_1(\x) = 2/3$ if $\x = (2,0)$ and $u_1(\x) = 1$ otherwise. 
Also for such $\x$, $u_2(\x) = 0$.
Buyer~$1$ then bids $b_1(\e_1) = 1 + 2/3 - 1 = 2/3$ at decision node~$(1,0)$  and bids $1$ at node~$(0,1)$.
On the other hand, because $u_2(\x) = 0$ for any $\x$ with $t = 1$, we have $u_2(\e_1) = 1/3$ and $u_2(\e_2) = 0$. 
So at decision node $(1,0)$, buyer~$1$ wins an 
item at price $1/3$, and at decision node $(0,1)$, it wins an item at price $0$. Therefore, $u_1(\e_1) = 4/3$, 
$u_1(\e_2) = 2$ and $u_2(\e_1) = u_2(\e_2) = 0$. Consequently, the bids at decision node~$\0$ are:
\begin{align*}
b_1(\0) & = 1 + 4/3 - 2 = 1/3 \\
b_2(\0) & = 1/3 + 0 - 0 = 1/3
\end{align*}
Breaking the tie in favour of buyer~$2$, this valuation profile has price of anarchy:
$$\frac{1/3+1+1}{3} = \frac{7}{9} = \frac{14}{18}$$
But this is greater than $13/18 + \epsilon/9$ for sufficiently small $\epsilon > 0$, hence it is not minimal. 
The problem seems to be in the induction step. The proofs assume in a minimal efficiency auction with $v_2(1) \leq b$ that if 
buyer~$2$ wins the first item, then the subauction starting from $(0,1)$ is also a minimal efficiency auction with $v_2(2) \leq v_2(1)$. 
Our example shows that this is not necessarily the case: the subauction starting from $(0,1)$ need not be minimal, as there are 
further constraints on it if buyer~$2$ is to win at $(0,0)$.

While the arguments in \cite{BBB09} may be incorrect, we strongly suspect that the claimed efficiency result is correct. Reducing the 
minimal efficiency problem to a mixed-integer linear program, we were able to verify for $T \leq 12$ that $T$-round minimal 
efficiency auctions with concave and non-decreasing valuations, where $v_1(T) \geq v_2(1)$ and buyer~$1$ is constrained to 
win $k$ items in the equilibrium, do have buyer~$2$'s incremental valuations satisfy equation~(\ref{eqn:minimalV2}) (with the same $k$).
This implies the $1-1/e$ efficiency bound holds for $T \leq 12$. 

\subsection{Social Welfare in Two-Buyer Sequential Auctions with No-Overbidding.}

Let's begin with the required definitions.
The {\em social welfare} of the allocation where buyer~$1$ wins $k\in [t] \cup \{0\}$ items from decision node~$\x$
is $\wel(k|\x) = \sum_{i=1}^k v_1(i|\x) + \sum_{i=1}^{t-k} v_2(i|\x)$. The {\em optimal social welfare}
is $\opt(\x) = \max_{k \in [t] \cup \{0\}} \wel(k | \x)$ and the corresponding {\em set of optimal allocations} from 
decision node~$\x$ 
is $\argopt(\x) = \arg  \max_{k \in [t] \cup \{0\}} \wel(k | \x)$.

The duopsony factor~(\ref{def:duopsony-factor}) of each buyer provides an explicit form of the set of optimal allocations.
Specifically, the set of optimal allocations at decision node $\x$ forms an interval.

\begin{claim}\label{cl:alloc-character}
The set of optimal allocations is $\argopt(\x) = [ f_1(\x), t-f_2(\x) ] \cap \mathbb{Z}$ at $\x$.
\end{claim}
\begin{proof}
To begin observe that, for $0 \leq k < t$, we have
\begin{equation}\label{eq:interval}
\wel(k+1|\x) - \wel(k|\x) = v_1(k+1|\x) - v_2(t-k|\x)
\end{equation}
First, assume that $k < f_1(\x)$ and so $k+1 \leq f_1(\x)$. By definition of duopsony factor~(\ref{def:duopsony-factor}) of buyer~$1$, it follows that 
$v_1(k+1|\x) - v_2(t-k|\x)>0$. So, by~(\ref{eq:interval}), we have $\wel(k+1|\x) > \wel(k|\x)$. 
The lower bound on the claimed interval then holds.

Second, assume $f_1(\x) \leq k < t-f_2(\x)$. Because $k+1 > f_1(\x)$, we have $v_1(k+1|\x) \le v_2(t-k|\x)$, by definition of 
duopsony factor~(\ref{def:duopsony-factor}) of buyer~$1$. Furthermore, as 
as $k < t-f_2(\x)$ we have $t-k > f_2(\x)$. Therefore, $v_2(t-k|\x) \le v_1(k+1|\x)$, by definition of duopsony factor~(\ref{def:duopsony-factor}) 
of buyer~$2$. Ergo, $v_2(t-k|\x) = v_1(k+1|\x)$. Thus $\wel(k+1|\x) = \wel(k|\x)$ and each integer in the interval
has the same social welfare.

Finally, assume $k \geq t-f_2(\x)$. Then $t-k \leq f_2(\x)$. Hence $v_2(t-k|\x) > v_1(k+1|\x)$
by definition of duopsony factor~(\ref{def:duopsony-factor}) of buyer~$2$.
It follows that $\wel(k+1|\x) < \wel(k|\x)$. The upper bound of the claimed interval then holds.
\end{proof}

Now, recall each equilibrium induces an equilibrium path in the extensive form game. Thus, we must study the social welfare arising along paths.
Formally, a {\em path} from decision node~$\x$ is a $(t+1)$-tuple $P = (\x^t,\x^{t-1},...,\x^1,\x^0)$, such that 
(i)~The first node on the path is $\x$, that is, $\x^t = \x$, and (ii)~Each successive node follows from some buyer~$j$ acquiring an item, that is,
for each $s \in [t]$ we have $\x^{s-1} = \x^s + \e_j$, for some buyer $j \in \{1,2\}$. 
We say that a path is an {\em equilibrium path} (taken with positive probability) if 
each successive node is the result of a buyer~$j$ winning with positive probability, that is
$\x^s + \e_j$ implies that $\pi_j(\be|\x) > 0$.
Finally, for any path $P = (\x^t,\x^{t-1},...,\x^1,\x^0)$ and any $s \in [t]$, we say that $P^s = (\x^s,...,\x^0)$ is the subpath of $P$ of length $s+1$.

We may then define a measure of efficiency along a path $P = (\x^t,\x^{t-1},...,\x^1,\x^0)$.
The {\em efficiency} along a path $P$ from $\x$ is:
$$\Gamma(P) = \frac{\wel(x^0_1 - x^t_1|\x)}{\opt(\x)}$$
We will see that if neither buyer has monopsony power then the efficiency of a path can be bounded by the 
efficiency along a subpath.

We now present a conditional lower bound on efficiency:

\begin{lemma}\label{lem:consec-bound}
Let $P$ be a path from~$\x$ where $t>1$. If neither buyer has monopsony power , that is $f_i(\x) < t$ for each buyer~$i$, 
then $\Gamma(P) \geq \Gamma(P^{t-1})$.
\end{lemma}
\begin{proof}
Without loss of generality, let $x^1_1 - x^0_1 = \e_1$. So buyer~$1$ is allocated the item at decision node $\x$ along the path $P$. 
Observe that:
\begin{align*}
\argopt(\x+\e_1) & = [f_1(\x+\e_1),t-1-f_2(\x+\e_1)] \\
& = [f_1(\x+\e_1),t-1-f_2(\x)]
\end{align*}
Here the first equality follows from Claim~\ref{cl:alloc-character} and the second equality follows from 
Claim~\ref{cl:sdf<t}. In particular, $t-1-f_2(\x)$ is in the interval $\argopt(\x+\e_1)$. 
But, again by Claim \ref{cl:alloc-character}, $t-f_2(\x)$ is in  the interval $\argopt(\x)$. Hence:
\begin{align}\label{eq:optxT}
\opt(\x) & = \sum_{i=1}^{t-f_2(\x)} v_1(i|\x) + \sum_{i=1}^{f_2(\x,t)} v_2(i|\x) \nonumber \\
& = v_1(1 |\x) + \sum_{i=2}^{t-f_2(\x)} v_1(i|\x) + \sum_{i=1}^{f_2(\x,t)} v_2(i|\x)\nonumber  \\
& = v_1(1 |\x) + \sum_{i=1}^{t-1-f_2(\x)} v_1(i|\x+\e_1) + \sum_{i=1}^{f_2(\x,t)} v_2(i|\x+\e_1)\nonumber  \\
& = v_1(1 |\x) + \opt(\x+\e_1)
\end{align}
Where the third equality follows from a change of variables in the first summation, and by considering the shifted 
valuations for $\x \rightarrow \x+\e_1$. We may then bound the efficiency along $P$:
\begin{align*}
\Gamma(P) & = \frac{\wel(x^0_1 - x^t_1|\x)}{\opt(\x)} \\
& = \frac{\sum_{i=1}^{x^0_1 - x^t_1} v_1(i|\x) + \sum_{i=1}^{x^0_2 - x^t_2} v_2(i|\x)}{\opt(\x)}\\
& = \frac{\sum_{i=2}^{x^0_1 - x^t_1} v_1(i|\x) + \sum_{i=1}^{x^0_2 - x^t_2} v_2(i|\x) + v_1(1|\x)}{\opt(\x)} \\
& = \frac{\sum_{i=1}^{x^0_1 - x^{t-1}_1} v_1(i|\x+\e_1) + \sum_{i=1}^{x^0_2 - x^{t-1}_2} v_2(i|\x+\e_1) + v_1(1|\x)}{\opt(\x)} \\
& = \frac{\wel(x^0_1 - x^{t-1}_1|\x+\e_1)+v_1(1|\x)}{\opt(\x+\e_1)+v_1(1|\x)} \\
& \geq \frac{\wel(x^0_1 - x^{t-1}_1|\x+\e_1)}{\opt(\x+\e_1)} \\
&= \Gamma(P^{t-1})
\end{align*}
Where the fourth equality follows since $x^{t-1}_1 = x_1^t+1$ and $x^{t-1}_2 = x^t_2$. The fifth equality holds by~(\ref{eq:optxT}). 
The inequality arises because $v_1(1|\x) \geq 0$.
\end{proof}

We may now provide a proof that a path $P$ that does not attain full efficiency has a subpath $P'$ whose initial node has a strict monopsonist, such that $\Gamma(P) \geq \Gamma(P')$.

\begin{lemma}\label{lem:path-efficiency}
Let $P$ be a path from~$\x$. If $f_j(\x^s) < t$ for each buyer $j\in\{1,2\}$ and for all $s \in [t]$ then $\Gamma(P) = 1$. 
Otherwise, if $f_j(\x^s) = s$ for some buyer $j$ and some $s \in [t]$ then 
$\Gamma(P) \geq \Gamma(P^{\bar{t}})$, where $\bar{t} = \max \{s \in [t]: \exists j \in\{1,2\}, f_j(\x^s) = s\}$.
\end{lemma}

\begin{proof}
First, suppose $f_j(\x^s) < s$ for each buyer $j\in\{1,2\}$ and for all $s \in [t]$. Then, by Lemma~\ref{lem:consec-bound}, 
 $\Gamma(P^s) \geq \Gamma(P^{s-1})$ for any $2\le s \in t$. Therefore, $\Gamma(P) \geq \Gamma(P^1)$. 
But, at $\x^1$, we have a single-item second price auction, which is fully efficient. That is, $\Gamma(P^1) = 1$.
Thus $\Gamma(P) = 1$.

Second, suppose $f_j(\x^s) = s$ for some buyer $j$ and some $s \in [t]$. By Claim~\ref{cl:sdf-t}, the set 
$\{s \in [t] :  \exists j \in \{1,2\}, f_j(\x^s) = s\}$ is 
equal to $[\bar{t}]$. In particular, for any buyer $j$ and any $s \in [t]\setminus[\bar{t}]$ we have $f_j(\x^s) < s$. 
Therefore, by repeated application of Lemma~\ref{lem:consec-bound}, it follows that
$\Gamma(P) \geq \Gamma(P^{\bar{t}})$.
\end{proof}

\begin{corollary}\label{cor:path-efficiency}
Let $P$ be an equilibrium path from $\x$. If $P$ does not have full efficiency, then for some $s \in [t]$, 
$\Gamma(P)$ is bounded below by $\Gamma(P^{s})$, where for some buyer~$j$, $f_j(\x^s) = s$.
\end{corollary}

So, to bound the efficiency of along an equilibrium path, we may restrict attention to equilibrium paths 
from a node $\x$ where, without loss of generality, buyer~$1$ has monopsony power.

\subsubsection{The Real Extension of the Incremental Valuation Function.}

So it suffices to evaluate the efficiency of a sequential auction where buyer~$1$ is a strict monopsonist at decision node $\0$ for some fixed $T$. 
By Theorem~\ref{thm:bid-threshold-main}, the quantity of items that buyer $1$ wins at equilibrium is at least its greedy 
demand $\kappa_1 := \kappa_1(\0)$. Observe that, for any $k > \kappa_1$, the welfare $V_1(k) + V_2(T-k)$ is greater than $V_1(\kappa_1) + V_2(T-\kappa_1)$. 
So we need to only lower bound the social welfare of the allocation $(\kappa_1,T-\kappa_1)$.

To do so, first consider extending 
incremental valuations to the real line, where $\forall \tau \in [0,t]$:
\begin{align}\label{def:real-extension}
\bar{v}_{1}(\tau|\x)  &= v_{1}(\lceil \tau \rceil \, |\x) \nonumber \\
\bar{v}_{2}(\tau|\x) &= v_{2}(\lfloor \tau+1 \rfloor\, |\x)
\end{align}

Also note that buyers' valuation functions also enjoy a real extension:
\begin{equation*}
\forall \ell \in \mathbb{R}, \bar{V}_i(\ell|\x) = \int_{0}^\ell \bar{v}_i(\tau | \x) \,d\tau
\end{equation*}

Now, buyer~$1$ obtains her greedy utility, $\mu_1(\0)$. We already know that this implies that this is the maximum of its greedy strategies 
for purchasing any integral amount. However, with the real extension, buyer~$1$'s greedy utility must also maximise her payoff amongst greedy 
strategies purchasing a fractional amount of items.

\begin{lemma}
Suppose that buyer~$1$ is a strict monopsonist at decision node $\x$. Then for any $k \in [0,t]$:
$$\mu_1(\x) \geq \bar{V}_1(k|\x) - k \cdot \bar{v}_2(t-k|\x)$$
\end{lemma}

\begin{proof}
Suppose not, i.e. there exists $k \in (0,t)$ such that:
\begin{equation}\label{eqn:boundContradict}
\mu_1(\x) < \bar{V}_1(k|\x) - k \cdot \bar{v}_2(t-k|\x)
\end{equation}
Note that, as buyer~$1$ is a strict monopsonist, for any $\ell \in [0,t]$: 
$$\bar{v}_1(\ell | \x) = v_1(\lceil \ell \rceil | \x) > v_2(t- \lceil \ell \rceil + 1 | \x) = \bar{v}_2(t - \ell |\x)$$
Therefore:
\begin{equation}\label{eqn:boundContradict2}
0 \leq \int_k^{\lceil k \rceil} \bar{v}_1(\ell|\x) \, d\ell - (\lceil k \rceil - k) v_2(t-\lceil k \rceil + 1 | \x)
\end{equation}
Adding together Equation \ref{eqn:boundContradict} and Equation \ref{eqn:boundContradict2}:
$$\mu_1(\x) < V_1(\lceil k \rceil) - \lceil k \rceil v_2(t - \lceil k \rceil + 1 | \x) = \bar\mu_1(\lceil k \rceil | \x)$$
Contradiction to the definition of $\mu_1(\x)$.
\end{proof}

By definition of $\mu_1$, as incremental valuations are non-negative, this yields an immediate bound of $\bar{v}_2$:

\begin{corollary}\label{cor:lb-v2}
Suppose that buyer~$1$ is a strict monopsonist at decision node $\x$. Then for any $k \in [0,t]$:
$$\bar{v}_2(\ell|\x) \geq \frac{1}{t-\ell} \int_{\kappa_1(\x)}^{t - \ell} \bar{v}_1(\tau| \x) \, d\tau$$
\end{corollary}

\subsubsection{The Price of Anarchy in Two-Buyer Sequential Auctions with No-Overbidding.}

Corollary~\ref{cor:lb-v2} implies that if a strict monopsonist buyer $1$ demands a small number of items, buyer $2$'s incremental valuations cannot also be small. Also, by our equilibrium characterisation we know that buyer $1$ also purchases at least $\kappa_1(\0)$ items. Using these facts, we are able to prove an upper bound on the price of anarchy. 

\begin{theorem}\label{thm:PoA-lower}
A two-buyer sequential multiunit auction with concave and non-decreasing valuations has price of anarchy at least $(1-1/e)$,
given no-overbidding. 
\end{theorem}
\begin{proof}
Take a realised equilibrium path: by Corollary \ref{cor:path-efficiency}, restrict attention to the case when buyer~$1$ is a strict monopsonist and when 
the equilibrium path starts from decision node $\0$ for some $T > 0$. Then 
social welfare is decreasing in the final allocation of buyer~$1$. By Theorem \ref{thm:bid-threshold-main}, buyer~$1$ wins at least $\kappa_1(\0)$ items.
Denoting $\kappa_1(\0) := \kappa$ for simplicity, the following is then a lower bound on the social welfare:
$$\int_0^\kappa \bar{v}_1(\tau) \, d\tau + \int_0^{T-\kappa} \bar{v}_2(\tau) \, d\tau$$
Then by Corollary \ref{cor:lb-v2}, we have a lower bound on social welfare:
\begin{align*}
\int_0^\kappa \bar{v}_1(\tau) \, d\tau + \int_0^{T-\kappa} \frac{1}{T-\tau} \int_\kappa^{T-\tau} \bar{v}_1(\lambda) \, d\lambda \, d\tau 
& = \int_0^\kappa \bar{v}_1(\tau) \, d\tau + \int_\kappa^{T} \int_0^{T-\lambda} \frac{1}{T-\tau} \bar{v}_1(\lambda) \, d\tau \, d\lambda \\
& = \int_0^\kappa \bar{v}_1(\tau) \, d\tau + \int_\kappa^{T} \ln \Big( \frac{T}{\lambda} \Big) \bar{v}_1(\lambda) \, d\lambda \\
& \geq \inf_{\gamma \in [0,T]} \int_0^\gamma\bar{v}_1(\tau) \, d\tau + \int_\gamma^{T} \ln \Big( \frac{T}{\lambda} \Big) \bar{v}_1(\lambda) \, d\lambda
\end{align*}
The expression in the infinimum is differentiable almost everywhere, in fact for any non-integral $\gamma$. Then evaluating the first-order condition for $\gamma$:
$$\bar{v}_1(\gamma) \left( 1-\ln \left( \frac{T}{\gamma} \right) \right) = 0 \Rightarrow \gamma = \frac{T}{e}$$
The function is indeed differentiable at $\gamma = T/e$, as $T/e$ is irrational (hence is not an integer). Minimality is ensured by a first derivative test. Plugging back this 
value to find the minimum:
\begin{align*}
\inf_{\gamma \in [0,T]} \int_0^\gamma\bar{v}_1(\tau) \, d\tau + \int_\gamma^{T} \ln \Big( \frac{T}{\lambda} \Big) \bar{v}_1(\lambda) \, d\lambda 
& = \bar{V}_1\Big( \frac{T}{e} \Big) + \int_{T/e}^{T} \ln \Big( \frac{T}{\lambda} \Big) \bar{v}_1(\lambda) \, d\lambda \\
& = \bar{V}_1\Big( \frac{T}{e} \Big) + \ln \Big( \frac{T}{\lambda} \Big) \bar{V}_1(\lambda) \Big|^{T}_{T/e} + \int_{T/e}^{T} \frac{\bar{V}_1(\lambda)}{\lambda} \, d\lambda \\
& = \int_{T/e}^{T} \frac{\bar{V}_1(\lambda)}{\lambda} \, d\lambda \geq \frac{\bar{V}_1(T)}{T} \cdot \Big( T - \frac{T}{e} \Big) = V_1(T) \cdot \Big( 1 - \frac{1}{e} \Big)
\end{align*}

Here, the second equality follows from integration by parts. The inequality then follows since $\bar{V}_1(\lambda)$ is normalised at zero and concave, which implies that $\bar{V}_1(\lambda) / \lambda$ is decreasing in $\lambda$. As $V_1(T)$ is the optimal social welfare, we conclude that resultant social welfare must be at least $(1-1/e)$ times that of an optimal allocation.
\end{proof}

To conclude our discussion, we provide a matching lower bound on the price of anarchy:

\begin{theorem}\label{thm:PoA-upper}
There exist two-buyer sequential multiunit auctions with concave and non-decreasing valuations and $T$ items, 
whose efficiency tends to $(1-1/e)$ as $T$ grows,
given no-overbidding. 
\end{theorem}
\begin{proof}
To see that the limit is tight as $T \rightarrow \infty$, consider the following valuation profile $\forall i \in [T]$:
\begin{align}\label{val-profile}
v_1(i) & = 1 \nonumber \\
v_2(i) & = \max \left\{ \frac{\lfloor T(1-1/e) \rfloor - i + 1}{T - i + 1}, 0 \right\}
\end{align}
Then $f_1(\0) = T$, $\mu_1(\0) = T/e$, and $\kappa_1(\0) \in \{ \lfloor T/e \rfloor, \lceil T/e \rceil \}$. Let all ties be broken in favour 
of buyer~$2$. By Theorem \ref{thm:bid-threshold-main}, buyer~$2$ wins $T-\kappa_1(\0)$ items and buyer~$1$ wins $\kappa_1(\0)$ items -- as 
by Lemma \ref{cl:top-marginal} buyer $2$'s greedy strategies 
will be weakly above those of buyer $1$'s until buyer $1$'s demand binds. 
Therefore, an upper bound on the social welfare at equilibrium is:
\begin{align*}
 \sum_{i = 1}^{\kappa_1(\0)} v_{1}(i) + \sum_{i = 1}^{T-\kappa_1(\0)} v_{2}(i) 
	       & \leq  \sum_{i = 1}^{\lceil T/e \rceil} v_{1}(i|\0) + \sum_{i = 1}^{\lfloor T(1-1/e) \rfloor} v_{2}(i|\0) \\
	       & = \lceil T/e \rceil + \sum_{i = 1}^{\lfloor T(1-1/e) \rfloor} 1 - \frac{T}{e}\cdot \frac{1}{T-i+1}\\
	       & = T - \sum_{i = 1}^{\lfloor T(1-1/e) \rfloor} \frac{T}{e} \cdot \frac{1}{T-i+1}
\end{align*} 
Here the first equality follows from~(\ref{val-profile}), and the fact that $\kappa_1(\0) \leq \lceil T/e \rceil$. 
Because $\opt(\0) = T$, it immediately follows that an upper bound on efficiency is:
$$ 1 - \frac{1}{e}\cdot \sum_{i = 1}^{\lfloor T(1-1/e) \rfloor} \frac{1}{T-i+1}$$
Taking the limit as $T \rightarrow \infty$, we obtain:
\begin{align*}
\lim_{T\rightarrow \infty} 1 - \frac{1}{e}\cdot \sum_{i = 1}^{\lfloor T(1-1/e) \rfloor} \frac{1}{T-i+1}
& = 1 - \lim_{T\rightarrow \infty} \frac{1}{e}\cdot \sum_{i = 1}^{\lfloor T(1-1/e) \rfloor} \frac{1}{T}\cdot  \frac{1}{1-i/T+1/T} \\
& = 1 - \frac{1}{e} \cdot \int_{0}^{1-1/e} \frac{1}{1-x} \, dx \\
& = 1 + \frac{1}{e} \cdot \ln(1-x) \bigg{|}_{x = 0}^{1-1/e} \\
& = 1 - \frac{1}{e}
\end{align*}
Here the second equality follows by interpreting the summation as a Riemann integral.
\end{proof}

\section*{Acknowledgements}
We are very grateful to Rakesh Vohra for discussions on this topic.

\newpage

\end{document}